%% file: l4dc.tex
\documentclass[12pt]{l4dc2022}

\title[Efficient and Robust Cubic Regularized Newton Method]{Escaping Saddle Points in Distributed Newton's Method with Communication Efficiency and Byzantine Resilience}
\usepackage{times}

\usepackage{algorithm}
\usepackage{algorithmic}

\usepackage{amsmath}
\usepackage{enumerate,enumitem}

\usepackage{amssymb}
\usepackage{comment}
\usepackage{mathtools}
\usepackage{graphicx}
\newtheorem{assumption}{Assumption}
\newtheorem{fact}{Fact}

\newcommand*\bH{\mathbf{H}}

\newcommand*\x{\mathbf{x}}
\newcommand*\s{\mathbf{s}}
\newcommand*\w{\mathbf{w}}
\newcommand*\y{\mathbf{y}}

\newcommand*\g{\mathbf{g}}

\newcommand*\bI{\mathbf{I}}

\newcommand*\cM{\mathcal{M}}

\newcommand*\cB{\mathcal{B}}
\newcommand*\cT{\mathcal{T}}
\newcommand*\cU{\mathcal{U}}

\usepackage[compact]{titlesec}
\titlespacing*{\subsection}{0pt}{0.25\baselineskip}{-0.1\baselineskip}
\titlespacing*{\paragraph}{0pt}{0.025\baselineskip}{0.5\baselineskip}

\author{%
 \Name{Avishek Ghosh} \Email{a2ghosh@ucsd.edu}\\
 \addr UCSD
 \AND
 \Name{Raj Kumar Maity} \Email{rajkmaity@cs.umass.edu}\\
 \addr UMass Amherst%
 \AND
 \Name{Arya Mazumdar} \Email{arya@ucsd.edu}\\
 \addr UCSD%
 \AND
 \Name{Kannan Ramachandran} \Email{kannanr@eecs.berkeley.edu}\\
 \addr UC Berkeley%
}

\begin{document}
\maketitle
\vspace{-8mm}
\begin{abstract}%
 The problem of saddle-point avoidance for non-convex optimization is quite challenging in large scale distributed learning frameworks, such as Federated Learning, especially in the presence of Byzantine workers. The celebrated cubic-regularized Newton method of  \cite{nest} is one of the most elegant ways to  avoid saddle-points in the standard centralized (non-distributed) setup. In this paper, we extend the cubic-regularized Newton method to a distributed framework and simultaneously address several practical challenges like communication bottleneck and Byzantine attacks. Note that the issue of saddle-point avoidance becomes more crucial in the presence of Byzantine machines since rogue machines may create \emph{fake local minima} near the saddle-points of the loss function, also known as the saddle-point attack. Being a second order algorithm, our iteration complexity is much lower than the first order counterparts. Furthermore we use compression (or sparsification) techniques like $\delta$-approximate compression for communication efficiency. We obtain theoretical guarantees for our proposed scheme under several settings including approximate (sub-sampled) gradients and Hessians. Moreover, we validate our theoretical findings with experiments using standard datasets and several types of Byzantine attacks, and obtain an improvement of $25\%$ with respect to first order methods in iteration complexity.

\end{abstract}


\input{intro_ai}

\input{problem_ai}

\input{byz_ai}
\input{expe_ai}


\newpage
\bibliography{icmlref}

\input{appendix5}
\input{appendixexp_ai}
\end{document}

%% file: intro_ai.tex
\section{Introduction} \label{sec:intro}
\vspace{ -5 pt}
Motivated by the real-world applications such as recommendation systems, image recognition, and conversational AI, it has become crucial to implement learning algorithms in a distributed fashion. In a commonly used framework, namely data-parallelism, large data-sets are distributed among several worker machines for parallel processing. In many applications, like Federated Learning (FL) \cite{federated},
data is stored in user devices such as mobile phones and personal computers. In a standard distributed framework, several worker machines perform local computations and communicate to the center machine (a parameter server), and the center machine aggregates and broadcasts the information iteratively.

In this setting, it is well-known that one of the major challenges is to tackle the behavior of the Byzantine machines \cite{lamport}. This can happen owing to software or hardware crashes, poor communication link between the worker and the center machine, stalled computations, and even co-ordinated or malicious attacks by a third party. In this setup, it is generally assumed (see \cite{dong1,blanchard2017byzantine} that a subset of worker machines behave completely arbitrarily even in a way that depends on the algorithm used and the data on the other machines, thereby capturing the unpredictable nature of the errors.

Another critical challenge in this distributed setup is the communication cost between the worker and the center machine. The gains we obtain by parellelization of the task among several worker machines often get bottle-necked by the communication cost between the worker and the center machine. In applications like Federated learning, this communication cost is directly linked with the (internet) bandwidth of the users and thus resource constrained. It is well known that in-terms of the number of iterations, second order methods (like Newton and its variants) outperform their competitor; the first order gradient based methods. In this work, we  simultaneously handle the Byzantine and communication cost aspects of distributed optimization for non-convex functions.

In this paper, we focus on optimizing a non-convex loss function $f(.)$ in a distributed optimization framework. We have $m$ worker machines, out of which $\alpha$ fraction may behave in a Byzantine fashion, where $\alpha < \frac{1}{2}$. Optimizing a loss function in a distributed setup has gained a lot of attention in recent years \cite{alistrah,blanchard2017byzantine,feng,chen}. However, most of these approaches either work when $f(.)$ is convex, or provide weak guarantees in the non-convex case (for  example: zero gradient points, maybe a saddle point). 

In order to fit complex machine learning models, one often requires to find local minima of a non-convex loss $f(.)$, instead of critical points only, which may include several saddle points. Training deep neural networks and other high-capacity learning architectures \cite{soudry2016bad,ge_etal} are some of the examples where finding local minima is crucial. \cite{ge_etal,kenji} shows that the stationary points of these problems are in fact saddle points and far away from any local minimum, and hence designing efficient algorithm that escapes saddle points is of interest. Moreover,  \cite{jain2017global,sun_etal} argue that saddle points can lead to highly sub-optimal solutions in many problems of interest. This is amplified in high dimension as shown in \cite{dauphin}, and becomes the main bottleneck in training deep neural nets. Furthermore, a line of recent work \cite{sun_etal,bhojanapalli2016global,Sun_2017}, shows that for many non-convex problems, it is sufficient to find a local minimum. In fact, in many problems of interest, all local minima are global minima (e.g., dictionary learning \cite{Sun_2017}, phase retrieval \cite{sun_etal}, matrix sensing and completion \cite{bhojanapalli2016global,ge_etal}, and some of neural nets \cite{kenji}). Also, in \cite{choromanska2015loss}, it is argued that for  more general neural nets, the local minima are as good as global minima.

The issue of local minima convergence becomes non-trivial in the presence of Byzantine workers. Since we do not assume anything on the behavior of the Byzantine workers, it is certainly conceivable that by appropriately modifying their messages to the center, they can create \emph{fake local minima} that are close to the saddle point of the loss function $f(.)$, and these are far away from the true local minima. This is popularly known as the \emph{saddle-point attack} (see \cite{dong}), and it can arbitrarily destroy the performance of any non-robust learning algorithm. Hence, our goal is to design an algorithm that escapes saddle points of $f(.)$ in an efficient manner as well as resists the saddle-point attack simultaneously. The complexity of such an algorithm emerges from the the interplay between non-convexity of the loss function and the behavior of the Byzantine machines.

The problem of saddle point avoidance in the context of non-convex optimization has received considerable attention in the past few years. In the seminal paper of \cite{pgd}, a (first order) gradient descent based approach is proposed. A few papers \cite{neon,neon2} following the above use various modifications to obtain saddle point avoidance guarantees. A Byzantine robust first order saddle point avoidance algorithm is proposed by Yin et al. \cite{dong}, and probably is the closest to this work. In \cite{dong}, the authors propose a repeated check-and-escape type of first order gradient descent based algorithm. First of all, being a first order algorithm, the convergence rate is quite slow (the rate for gradient decay is $1/\sqrt{T}$, where $T$ is the number of iterations). Moreover, implementation-wise, the algorithm presented in \cite{dong} is computation heavy, and takes potentially many iterations between the center and the worker machines. Hence, this algorithm is not efficient in terms of the communication cost.

In this work, we consider a variation of the famous cubic-regularized Newton algorithm of Nesterov and Polyak  \cite{nest}, which efficiently escapes the saddle points of a non-convex function by appropriately choosing a regularization and thus pushing the Hessian towards a positive semi-definite matrix. The primary motivation behind this choice is the faster convergence rate compared to first order methods, which is crucial in terms of communication efficiency in applications like Federated Learning. Indeed, the rate of gradient decay is $\frac{1}{T^{2/3}}$.

We consider a distributed variant of the cubic regularized Newton algorithm. In this scheme, the center machine asks the workers to solve an auxiliary function and return the result. Note that the complexity of the problem is partially transferred to the worker machines. It is worth mentioning that in most distributed optimization paradigm, including Federated Learning, the workers posses sufficient compute power to handle this partial transfer of compute load, and in most cases, this is desirable \cite{federated}. The center machine aggregates the solution of the worker machines and takes a descent step. Note that, unlike gradient aggregation, the aggregation of the solutions of the local optimization problems is a highly non-linear operation. Hence, it is quite non-trivial to extend the centralized cubic regularized algorithm to a distributed one. The solution to the cubic regularization even lacks a closed form solution unlike the second order Hessian based update or the first order gradient based update. The analysis is carried out by leveraging the first order and  second order stationary conditions of the auxiliary function solved in each worker machines.

In addition to this, we simulateneously use (i) a $\delta$-approximate compressor (defined shortly) to compress the message send from workers to center to gain further communication reduction and (ii) a simple norm-based thresholding to robustify against adversarial attacks. Norm based thresholding is a standard trick for Byzantine resilience as featured in \cite{ghosh2020communication,ghosh2020distributed}. However, since the local optimization problem lacks a closed form solution, using norm-based trimming is also technical challenging in this case. We now list our contributions.

\subsection{Our Contributions}
We propose a novel distributed and robust cubic regularized Newton algorithm, that escapes saddle point efficiently. We prove that the algorithm convergence at a rate of $\frac{1}{T^{2/3}}$, which is faster than the first order methods (which converge at $1/\sqrt{T}$ rate, see \cite{dong}). Also, the rate matches to the centralized scheme of \cite{nest} and hence, we do not lose in terms of convergence rate while making the algorithm distributed. We emphasis that since the center machine aggregates solutions of local (auxiliary) loss functions, this extension is quite non-trivial and technically challenging. This fast convergence reduces the number of iterations (and hence the communication cost) required to achieve a target accuracy.

Along with the saddle point avoidance, we simultaneously address the issues of (i) communication efficiency and (ii) Byzantine resilience by using a $\delta$-approximate compressor and a  norm-based thresholding scheme respectively. A major technical challenge here is to simultaneously address the above mentioned issues, and it turns out that with a proper parameter choice (step size etc.) it is possible to carry out the analysis jointly.  

In Section~\ref{sec:experiments}, we verify our theoretical findings via experiments. We use benchmark LIBSVM (\cite{libsvm}) datasets for logistic regression and non-convex robust regression and show convergence results for both non-Byzantine and several different Byzantine attacks. Specifically, we characterize the total iteration complexity (defined in Section~\ref{sec:experiments}) of  our algorithm, and compare it with the first order methods. We observe that the algorithm of \cite{dong} requires $25\%$ more total iterations than ours.
\paragraph{Preliminaries:} A point $\x$ is said to satisfy the $\epsilon$-second order stationary condition of $f(.)$ if, 
\vspace{-2mm}
\begin{align*}
\|\nabla f(\x) \| \leq \epsilon \qquad \lambda_{\text{min}}(\nabla^2 f(\x) ) \geq -\sqrt{\epsilon}.
\vspace{-2mm}
\end{align*}
$\nabla f(\x)$ denotes the gradient of the function and $\lambda_{\text{min}}(\nabla^2 f(\x))$ denotes the minimum eigenvalue of the Hessian of the function. Hence, under the assumption (which is standard in the literature, see \cite{pgd,dong}) that all saddle points are strict (i.e., $\lambda_{\min}(\nabla^2 f(\x_s)) <0$ for any saddle point $\x_s$), all second order stationary points (with $\epsilon=0$) are local minima, and hence converging to a stationary point is equivalent to converging to a local minima.

\subsection{Problem Formulation}\label{sec:prob}
We minimize a loss function of the form: $f(\x)= \frac{1}{m}\sum_{i=1}^m f_i(\x)$, where the function $f : \mathbb{R}^d \rightarrow \mathbb{R}$ is twice differentiable and  non-convex.  In this work, we consider distributed optimization framework with $m$ worker machines and one center machine where the worker machines communicate to the center machine. Each worker machine is associated with a local loss function $f_i $. We assume that the data distribution is non-iid across workers, which is standard in frameworks like FL. In addition to this, we also consider the case where $\alpha$ fraction of the worker machines are Byzantine for some $\alpha < \frac12 $.  The Byzantine machines can send arbitrary updates to the central machine which can disrupt the learning. Furthermore, the Byzantine machines can collude with each other, create \emph{fake local minima} or attack  maliciously by gaining information about the learning algorithm and other workers. 
In the rest of the paper, the norm $\|\cdot\|$ will refer to $\ell_2$ norm or spectral norm when the argument is a vector or a matrix respectively. 

Next, we consider a generic class of compressors from \cite{errorfeed}:
\vspace{-2mm}
\begin{definition}[$\delta$-Approximate Compressor]
\label{def:compress}
An operator $Q(.): \mathbb{R}^d \rightarrow \mathbb{R}^d$ is defined as $\delta$ approximate compressor on a set $\mathcal{S} \subseteq \mathbb{R}^d$ if,  $\forall \, x \in \mathcal{S}$, $\|Q(x) -x \|^2 \leq (1-\delta) \|x\|^2$,
where $\delta \in (0,1]$ is the compression factor.
\end{definition} 
\vspace{-2mm}
\noindent Furthermore, a randomized operator $Q(.)$ is $\delta$-approximate compressor on a set $\mathcal{S} \subseteq \mathbb{R}^d$ if, 
the above holds on expectation. In this paper, for the clarity of exposition, we consider the deterministic form of the compressor (as in Definition~\ref{def:compress}). However, the results can be easily extended for randomized $Q(.)$. Notice that $\delta=1$ implies $Q(x)=x$ (no compression).

\section{Related Work}
\paragraph*{Saddle Point avoidance algorithms:}
In the recent years, there are handful first order algorithms \cite{lee2016gradient,lee2017first,du2017gradient} that focus on the escaping saddle points and convergence to local minima. The critical algorithmic aspect is  running gradient based algorithm and adding perturbation to the iterates when the gradient is small. ByzantinePGD \cite{dong}, PGD \cite{pgd}, Neon+GD\cite{neon}, Neon2+GD \cite{neon2} are examples of such algorithms. The work of Nesterov and Polyak \cite{nest} first proposes the cubic regularized second order Newton method and provides analysis for the second order stationary condition.  An algorithm called Adaptive Regularization with Cubics (ARC) was developed by \cite{cartis,cartis1}  where cubic regularized Newton method with access to inexact Hessian was studied.  Cubic regularization with both the gradient and Hessian being inexact was studied in \cite{stoch}. In \cite{subsample}, a  cubic regularized Newton with  sub-sampled Hessian and gradient was proposed and analyzed. Momentum based cubic regularized algorithm was studied in \cite{momentcubic}. A variance reduced cubic regularized algorithm was proposed in \cite{zhou2018stochastic,wang2019stochastic}. In terms of solving the cubic sub-problem, \cite{carmon} proposes a gradient based algorithm  and \cite{agarwal2017finding} provides a Hessian-vector product technique.

\paragraph*{Compression: } In the recent years, several gradient quantization or sparsification schemes have been studied in \cite{vqsgd, alistarh2017communication,alistarh2018convergence,atomo, terngrad, qsgd, ivkin2019communication}. In \cite{errorfeed}, the authors introduced the idea of $\delta$-approximate compressor. In \cite{ghosh2020comm}, the authors used $\delta$-approximate compressor to sparsify the second order update.  
 
\paragraph*{Byzantine resilience:} 
  The effect of  adversaries on convergence of non-convex optimization was studied in \cite{Damaskinos,mhamdi2018hidden}. In the distributed learning context, \cite{feng} proposes  one shot median based robust learning.  A median of mean based algorithm was proposed in  \cite{chen} where the worker machines are grouped in batches and the Byzantine resilience is achieved by computing the median of the grouped machines. Later \cite{dong1} proposes co-ordinate wise  median, trimmed mean and iterative filtering  based approaches. Communication-efficient and Byzantine robust algorithms were developed in \cite{anima,ghosh2020communication}.  A norm based thresholding approach for Byzantine resilience for distributed Newton algorithm was also developed~\cite{ghosh2020distributed}. All these works provide only first order convergence guarantee (small gradient). The work \cite{dong} is the only one that  provides second order guarantee (Hessian positive semi-definite) under Byzantine attack. 

%% file: problem_ai.tex
\begin{algorithm}[h!]
  \caption{Byzantine Robust Distributed Cubic Regularized Newton Algorithm}
  \begin{algorithmic}[1]
 \STATE  \textbf{Input:} Step size $\eta_k$, parameter $0\leq \alpha \le \beta, \gamma>0,M>0$ and $\delta$-approximate  compressor $Q$. 
 \STATE \textbf{Initialize:} Initial iterate $\x_0 \in \mathbb{R}^d$ \\
  \FOR{$k=0,1, \ldots, T-1 $}
\STATE \underline{Central machine:} broadcasts $\x_k$  \\

  \textbf{ for $ i \in [m]$ do in parallel}\\
  \STATE \underline{$i$-th worker machine:} \\
    \textit{ Non-Byzantine:} Compute local gradient $\g_{i,k}$  and Hessian $\bH_{i,k}$; locally solve the problem equation~\eqref{eq:sub}. Use the compressor $Q$  and  send $Q(\s_{i,k+1})$ to the center,\\
     \textit{ Byzantine}: Generate $\star$ (arbitrary), and send it to the center machine \\
    \textbf{end for}
\STATE \underline{Center Machine:} \\
    (i) Sort the worker machines in a non decreasing order according to norm of updates $\{ Q(\s_{i,k+1})\}_{i=1}^m$ from the local machines \\ 
     (ii) Return the indices of the first $1-\beta$ fraction of machines as $\mathcal{U}_t$, \\
     (iii) Update  parameter: $ \x_{k+1}= \x_k + \eta_k\frac{1}{|\mathcal{U}_t|}\sum_{i\in \mathcal{U}_t} Q(\s_{i,k+1})$
  \ENDFOR
  \end{algorithmic}\label{alg:main_algo}
\end{algorithm}

\section{Compression, Byzantine Resilience and Distributed Cubic Regularized Newton}

In this section, we describe a communication efficient and Byzantine robust distributed cubic Newton algorithm that can avoid saddle point by meeting the condition second order stationary point and thus converge to a local minima for non-convex loss function. Starting with initialization $\x_0$, the center machine broadcasts the parameter to the worker machines. At $k$-th iteration, the $i$-th worker machine  solves a cubic-regularized auxiliary loss function based on its local data:
\vspace{-1mm}
\begin{align}
 \s_{i,k+1}= \arg\min_{\s} \g_{i,k}^T \s + \frac{\gamma}{2}\s^T\bH_{i,k}\s + \frac{M}{6}\gamma^2\|\s\|^3, \label{eq:sub}
 \end{align} 
 \vspace{-1mm}
 where $M>0 , \gamma>0$ are parameter choose suitably  and $\g_{i,t},\bH_{i,t}$ are the gradient and Hessian of the local loss function $f_i$ computed on  data $(S_i)$ stored in the worker machine.
\begin{align*}
 \g_{i,k}  =\nabla f_i(x_k)= \frac{1}{|S_i|}\sum_{z_i \in S_i }\nabla f_i(x_k,z_i), \,
 \bH_{i,k}  =\nabla^2 f_i(x_k) = \frac{1}{|S_i|}\sum_{z_i \in S_i }\nabla^2 f_i(x_k,z_i).
\end{align*} 
After solving the problem described in \eqref{eq:sub}, each worker machine applies compression operator $Q$  as defined in Definition~\ref{def:compress}  on update $\s_{i,k+1}$. The application  of the compression on the update is to minimize the communication cost. 

Moreover, we also consider that  $\alpha (< \frac12)$ fraction  of the worker machines are Byzantine in nature. We denote the set of Byzantine worker machines by $\cB$  and the set of the rest of the good machines as $\cM$. In each iteration, the good machines send the  compressed update of solution of  the sub-cubic problem described in equation ~\eqref{eq:sub} and the Byzantine machines can send any arbitrary values or intentionally disrupt the learning algorithm with malicious updates. Moreover, in the non-convex optimization  problems, one of the more complicated and important issue is to avoid saddle points which can yield highly sub-optimal results. In the presence of Byzantine worker machines, they  can be in cohort to create a \emph{fake local minima} and drive the algorithm into sub-optimal region. Lack of any robust measure towards these type of intentional and unintentional attacks can be catastrophic to the learning procedure as the learning algorithm can get stuck in such sub-optimal point. To tackle such Byzantine worker machines,  we employ a simple process called \emph{norm based thresholding}.

After receiving all the updates from the worker machines, the central machine outputs a set $\cU $ which consists of the indexes of the worker machines with smallest norm. We choose the size of the set $\cU$ to be $(1-\beta)m$. Hence, we `trim' $\beta$ fraction of the worker machine so that we can control the iterated update by not letting the worker machines with large norm participate  and diverge the learning process. We denote the set of trimmed machine as $\cT$. We choose $\beta>\alpha$ so that at least one of the good machines gets trimmed. In this way, the norm of the all the updates in the set $\cU $ is bounded by at least the largest norm of the  good machines.

Notice that the $\delta$-approximate compressor is used to minimize the cost of communication and the norm based thresholding is used to mitigate the effect of the Byzantine machines.   The central machine updates the parameter, with step-size $\eta_k$ as $
\x_{k+1}= \x_k + \eta_k\frac{1}{|\mathcal{U}_t|}\sum_{i\in \mathcal{U}_t} Q(\s_{i,k+1}) $.
\begin{remark}\label{rmk:lam}
Note that, we introduce the parameter $\gamma$ in the cubic regularized sub-problem. The parameter $\gamma$ emphasizes the effect of the second and third order terms in the sub-problem. The choice of  $\gamma$ plays an important role in our analysis in handling the non-linear update from different worker machines. Such non-linearity vanishes if we choose $\gamma=0$.
\end{remark}
\subsection{Theoretical Guarantees}
We have the following standard assumptions:
\vspace{-3mm}
\begin{assumption}\label{asm:fun}
The non-convex loss function $f(.)$ is  twice continuously-differentiable and bounded below, i.e., $f^* = \inf_{\x \in \mathbb{R}^d} f(x)> -\infty $.
\end{assumption}
\vspace{-4mm}
\begin{assumption} \label{asm:lip}
The loss $f(.)$ is \emph{$L$-Lipschitz continuous} ($ \forall \x,\y$, $\left| f(\x)-f(\y) \right| \leq L\| \x -\y\|$), has \emph{$L_1$-Lipschitz gradients }    
   $(\left\| \nabla f(\x)- \nabla f(\y) \right\| \leq L_1\| \x -\y\|)$ and \emph{$L_2$-Lipschitz Hessian }  
  
  $(\left\| \nabla ^2 f(\x)- \nabla^2 f(\y) \right\|   \leq L_2\| \x -\y\|)$.
\end{assumption}
The above assumption states that the loss and the gradient and Hessian of the loss do not drastically change in the local neighborhood. These assumptions are standard in the analysis of the saddle point escape for cubic regularization (see \cite{stoch,subsample,nest,carmon}).

We assume the data distribution across workers to be non-iid. However, we assume that the local gradient and Hessian computed at worker machines (using local data) satisfies the following gradient and Hessian dissimilarity conditions. Note that these conditions are only applicable for non-Byzantine machines only. Byzantine machines do not adhere to any assumptions.
\begin{assumption} (Gradient dissimilarity) \label{asm:gradinexec}
For $\epsilon_g >0 $, we have $\| \nabla f(\x_k) - \g_{i,k}\| \leq \epsilon_g$ for all $k,i$.
\end{assumption}
\vspace{-4mm}
\begin{assumption}(Hessian dissimilarity) \label{asm:hessinexec}
For $\epsilon_H >0 $, we have $\| \nabla^2 f(\x_k) - \bH_{i,k}\| \leq \epsilon_H$ for all $k,i$.
\end{assumption}
\vspace{-2mm}
 Similar assumptions featured in previous literature. For example, in  \cite{karimireddy2020scaffold,fallah2020personalized}, the authors use similar kind of dissimilarity assumptions that are prevalent in the \emph{Federated learning} setup to highlight the \emph{non-iid} or heterogeneity  of the data among users.
\begin{remark}\label{rmk:iiddata}
[Values of $\epsilon_g$ and $\epsilon_H$ in special cases] Compared to the Assumptions ~\ref{asm:gradinexec}  and ~\ref{asm:hessinexec}, the gradient and Hessian bound have been studied under more relaxed condition. In \cite{subsample,stoch,momentcubic}, the authors consider gradient and Hessian  with sub-sampled data being drawn uniformly randomly from the data set. Due to the data being drawn in iid manner, both the bound $(\epsilon_g, \epsilon_H)$ parameters value  diminish at the rate $\propto 1/ \sqrt{|S|}$ where $|S|$ is the size of the data sample in each worker machine.  In \cite{ghosh2020distributed}, the authors analyze the deviation in case of \emph{data partitioning} where each worker node sample data uniformly \emph{without replacement} from a given data set.  
\end{remark}

\begin{theorem}\label{thm:oneround}
Suppose $0\leq \alpha \le \beta \leq \frac 12$ and Assumptions ~\ref{asm:fun}-\ref{asm:hessinexec} hold, and we choose  $M = \mathcal{O}( m (1+\sqrt{1-\delta})^3)$, $\eta = \gamma = c/T$. Then, after $T$ iterations of Algorithm~\ref{alg:main_algo}, the sequence $\{\x_i\}_{i=1}^T$ generated contains a point $\tilde{x}$ such that \begin{align}
\| \nabla f(\tilde{x})\| \leq &\frac{\chi_1}{T^{2/3}} + \epsilon_g + \mathcal{O}(1/T), \quad
 \lambda_{\min}\left(\nabla^2 f(\tilde{x}) \right) \geq  -\frac{\chi_2}{T^{\frac{1}{3}}} -\epsilon_H -\mathcal{O}(1/T), \\
\text{ where,    }\quad   \chi_1 &= \mathcal{O}([ \frac{(1-\alpha)(1+\sqrt{1-\delta})^2}{2(1-\beta)} +m (1+\sqrt{1-\delta})^3]) \nonumber  \\  \chi_2 &=\mathcal{O}([m (1+\sqrt{1-\delta})^3 +  \frac{(1+\sqrt{1-\delta})(1-\alpha)}{(1-\beta)}]).\nonumber\end{align}
\end{theorem}
\vspace{-4mm}
\begin{remark}\label{rmk:steps}
Both the gradient and the minimum eigenvalue of the Hessian in the Theorem~\ref{thm:oneround} have two parts. The first part decreases with the number iterations $T$. The gradient and the minimum eigenvalue of the Hessian have the rate of $O(1/T^{\frac23})$ and  $O(1/T^{\frac13})$, respectively. Both of these rates match the rates of the centralized version of the cubic regularized Newton. In the second parts of  the gradient bound and the minimum eigenvalue of the Hessian  have terms with $\epsilon_g,\epsilon_H$  factor. As mentioned above (see Remark   ~\ref{rmk:iiddata} ), in the special cases, both the terms $\epsilon_g$ and $\epsilon_H$ decrease at the rate of $1/\sqrt{|S|}$, where $|S|$ is the number of data in each of the worker machines.
\end{remark}
\begin{remark}\label{rmk:dong}
 [Comparison with \cite{dong}] In a recent work, \cite{dong} provides a \emph{perturbed gradient based algorithm } to escape the saddle point in  non-convex optimization in the presence of Byzantine worker machines. Also, in that paper, the Byzantine resilience   is achieved using  techniques such as trimmed mean, median and collaborative filtering. These methods require additional  assumptions (coordinate of the gradient being sub-exponential etc.) for the purpose of analysis. In this work, we do not require such assumptions. Moreover, we perform a simple \emph{norm based thresholding } that provides robustness. Also the perturbed gradient descent (PGD) actually requires multiple rounds of communications between the central machine and the worker machines whenever the norm of the gradient  is small as this is an indication of either a local minima or a saddle point. In contrast to that, our method does not require any additional communication for \emph{escaping} the saddle points. Our method provides such ability by virtue of cubic regularization.
\end{remark}
\begin{remark} 
Since our algorithm is second order in nature, it requires less number of iterations compared to the first order gradient based algorithms. Our algorithm achieves a superior  rate of $O(1/T^{\frac23})$ compared to  the gradient based approach of rate $O(1/\sqrt{T})$. Our algorithm dominates ByzantinePGD \cite{dong} in terms of convergence, communication rounds and simplicity. 
\end{remark}
\vspace{-2mm}
 We now state two corollaries.  First, we relax the condition of compression by choosing $\delta=1$.  The worker machines communicate the actual update instead of the compressed the update.
\begin{corollary}[No compression]
Suppose $0\leq \alpha \le \beta \leq \frac 12$ and Assumptions ~\ref{asm:fun}-\ref{asm:hessinexec} hold, and we choose  $M = \mathcal{O}( m )$, $\eta = \gamma = c/T$. Then, after $T$ iterations of Algorithm~\ref{alg:main_algo} for uncompressed update $(\delta=1)$, the sequence $\{\x_i\}_{i=1}^T$ generated contains a point $\tilde{x}$ such that \begin{align}
\| \nabla f(\tilde{x})\| \leq \frac{\chi_1}{T^{2/3}} + \epsilon_g + \mathcal{O}(1/T), \quad
 \lambda_{\min}\left(\nabla^2 f(\tilde{x}) \right) \geq  -\frac{\chi_2}{T^{\frac{1}{3}}} -\epsilon_H -\mathcal{O}(1/T), 
\end{align}
where, $\chi_1 =\chi_2= \mathcal{O}([ \frac{(1-\alpha)}{(1-\beta)} +m ]).$
\end{corollary}
 Next, we choose the non-Byzantine setup with $\alpha=\beta=0$ in addition to the uncompressed update. This is just the distributed variant of the cubic regularized Newton method.
\begin{corollary} [Non Byzantine and no compression]\label{cor:nonbyz}
Suppose Assumptions ~\ref{asm:fun}-\ref{asm:hessinexec} hold, and we choose  $M = \mathcal{O}( m )$, $\eta = \gamma = c/T$. Then, after $T$ iterations of Algorithm~\ref{alg:main_algo} for uncompressed update $(\delta=1)$, the sequence $\{\x_i\}_{i=1}^T$ generated contains a point $\tilde{x}$ such that \begin{align}
\| \nabla f(\tilde{x})\| \leq \frac{\chi_1}{T^{2/3}} + \epsilon_g + \mathcal{O}(1/T), \quad
 \lambda_{\min}\left(\nabla^2 f(\tilde{x}) \right) \geq  -\frac{\chi_2}{T^{\frac{1}{3}}} -\epsilon_H -\mathcal{O}(1/T), 
\end{align}
where, $\chi_1 =\chi_2= \mathcal{O}(m).$
\end{corollary}

\begin{remark}[Two rounds of communication $\epsilon_g=0$]\label{rmk:tworound}
We can improve the bound in the Corollary ~\ref{cor:nonbyz}, with the calculation of the actual gradient which requires one more round of communication in each iteration. In the first iteration, all the worker machines compute the gradient based on the stored data and send it to the center machine. The center machine averages them and then broadcast the global gradient $\nabla f(\x_k)= \frac{1}{m}\sum_{i=1}^m \g_{i,k} $ at iteration $k$. In this manner, the worker machines solve the sub-problem \eqref{eq:sub} with  the actual gradient.  This improves the gradient bound while the communication  remains $O(d)$ in each iteration.
\end{remark}
\vspace{-2mm}
\subsection{Solution of the cubic sub-problem}
The cubic regularized sub-problem ~\eqref{eq:sub} needs to be  solved to update the parameter. As this particular problem does not have a closed form solution, a solver is usually employed which yields a satisfactory  solution. In previous works, different types of solvers have been used. \cite{cartis,cartis1} solve the sub-problem using Lanczos based method in Krylov subspace. In  \cite{agarwal2017finding}, the authors propose a solver based on  Hessian-vector product and binary search. Gradient descent based solver is proposed in \cite{carmon,stoch}.

Previous works, \cite{momentcubic,zhou2018stochastic,wang2019stochastic}, consider the exact solution of the cubic sub-problem for theoretical analysis. Recently, inexact solutions to the sub-problem is also proposed in the centralized (non-distributed) framework. For instance, \cite{subsample} analyzes the cubic model with sub-sampled Hessian with approximate model minimization technique developed in \cite{cartis}. Moreover, \cite{stoch} shows improved analysis with gradient based minimization which is a variant studied in \cite{carmon}. Both exact and inexact solutions to the sub-problem yields similar theoretical guarantees.

In our framework, each worker machine is tasked with solving the sub-problem. For the purpose of theoretical convergence analysis, we consider that woker machines obtain the exact solution in each round. However, in experiments (Section ~\ref{sec:experiments}), we apply the gradient based solver of \cite{stoch} to solve the sub-problem. Here, we let each worker machines run the gradient based solver for $10$ iterations and send the update to the center machine in each iteration. 

%% file: expe_ai.tex
\section{Experimental Results}
\label{sec:experiments}
First we show that our algorithm indeed \emph{escapes saddle point} with a toy example. We choose a  $2$ dimensional  example: $\min_{w \in \mathbb{R}^2} [ f_1(w)+f_2(w)]$ where $f_1(w)=w_1^2-w_2^2 $ and $f_2(w)=2w_1^2-2w_2^2$ (Here $w_i^2$ denotes the $i$-th coordinate of $w^2$. This problem is the sum of two non-convex function and  has a saddle point at $(0,0)$. In Figure ~\ref{fig:saddle} (left most)  we observe  that our algorithm escapes the saddle point $(0,0)$, with random initialization.

Next, we validate our algorithm  on benchmark LIBSVM (\cite{libsvm}) data-set in both convex and non-convex problems. We choose the following loss functions: (a) Logistic loss:$
\min_{\w \in \mathbb{R}^d} \frac{1}{n}\sum_i \log \left(1 + \exp(-y_i\x_i^T\w) \right) + \frac{\lambda}{2n}\|\w\|^2$, and (b) Non-convex robust linear regression:
$
\min_{\w \in \mathbb{R}^d} \frac{1}{n}\sum_i \log \left( \frac{(y_i -\w^T\x_i)^2}{2} +1 \right)$, where $\w \in \mathbb{R}^d$ is the parameter, $\{\x_i\}_{i=1}^n \in \mathbb{R}^d$ are  the feature vectors and $\{\y_i\}_{i=1}^n \in \{0,1\}$ are the corresponding labels. We choose `a9a'($d=123,n \approx 32K$, we split the data into $70/30$ and use as training/testing purpose) and `w8a'(training data $d=300,n \approx 50K$ and testing data $d=300,n \approx 15K$ ) classification datasets and partition the data in 20 different worker machines. 

First, we demonstrate Algorithm~\ref{alg:main_algo} in the presence of Byzantine machines and compressed update. For compression, each worker applies compression operator of QSGD \cite{qsgd}. For a given vector $\x \in \mathbb{R}^d, [Q(\x)]_i=\|\x\|_2 \text{sign}(\x_i)\times \text{Ber}(|\x_i|/\|x\|_2)$ for all $i \in [d]$. In this work, we consider the following four Byzantine attacks: 
(1) `Gaussian Noise attack': where the Byzantine worker machines add Gaussian noise to the update.
 (2) `Random label attack': where the Byzantine worker machines train and learn based on random labels instead of the proper labels.
  (3) `Flipped label attack': where (for Binary classification) the Byzantine worker machines flip the labels of the data and learn based on wrong labels.
  (4) `Negative update attack': where the Byzantine workers computes the update $\s$ (here solves the sub-problem in Eq.~\eqref{eq:sub}) and communicates $-c*\s$ with $c\in (0,1)$ making the direction of the update opposite of the actual one.
In Figure ~\ref{fig:comprobustbyz}, we plot the function value of the robust linear regression problem for 'flipped labels` and 'negative update` attacks with compressed update for both `w8a' and `a9a' dataset. We choose the parameters $\lambda=1, M=10$, learning rate $\eta_k=1$, $\alpha=\{.1,.15,.2\}$ and $\beta= \alpha + \frac{2}{m}$, where number of worker machines $m=20$. The results with Gaussian and random labels attacks are shown in Appendix.
\begin{figure}[t!]%
\centering
\subfigure[]{
\includegraphics[height = 3.2cm,width=3.5cm]{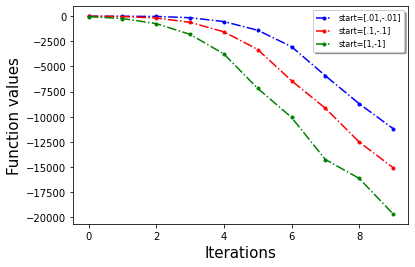}}%
\subfigure[]{
\includegraphics[height = 3.2cm,width=3.5cm]{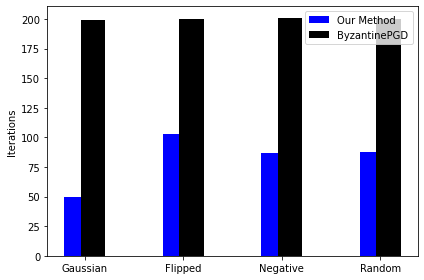}}%
\subfigure[]{
\includegraphics[height = 3.2cm,width=3.5cm]{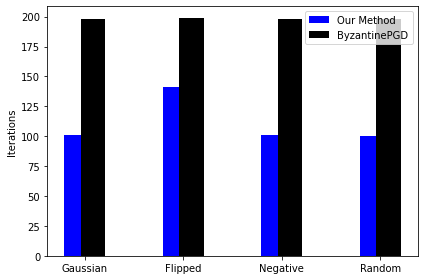}}%
\subfigure[]{
\includegraphics[height = 3.2cm,width=3.5cm]{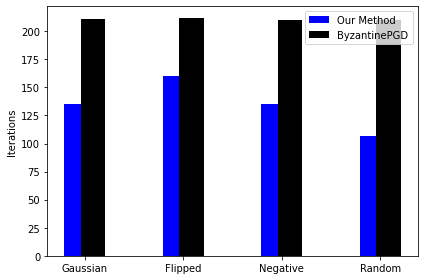}}%
\vspace{-3mm}
\caption{  (a) Plot of the function value with different initialization to show that the algorithm escapes the saddle point with functional value $0$. Comparison of our algorithm with ByzantinePGD \cite{dong} in terms of  the total number of iterations.}
\label{fig:saddle}
\vspace{-9mm}
\end{figure}
\begin{figure}[h]%
\vspace{-1mm}
\centering
\subfigure[]{
\includegraphics[height = 3.2cm,width=3.5cm]{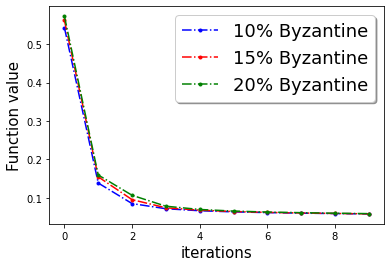}}%
\subfigure[]{
\includegraphics[height = 3.2cm,width=3.5cm]{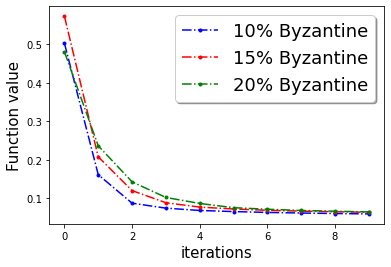}}%
\subfigure[]{
\includegraphics[height = 3.2cm,width=3.5cm]{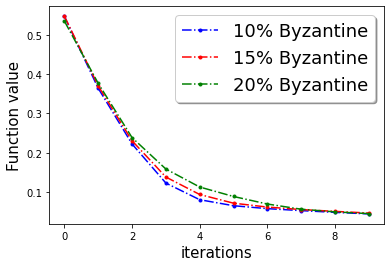}}%
\subfigure[]{
\includegraphics[height = 3.2cm,width=3.5cm]{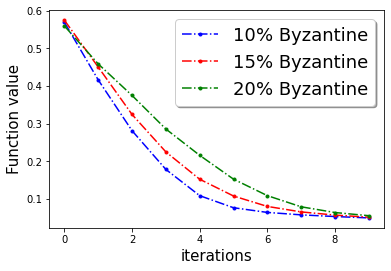}}%
\vspace{-4mm}
\caption{Loss  of the training data `a9a' (first row) and `w8a' (second row) with $10\%,15\%,20\% $ Byzantine worker machines for (a,c). Flipped label attack.(b,d). Negative Update attack.}%
    \label{fig:comprobustbyz}%
    \vspace{-8mm}
\end{figure}
\paragraph{Comparison with \cite{dong}}: We compare our uncompressed version of the algorithm $(\delta=1)$  with ByzantinePGD of \cite{dong}. We take the \emph{total number of iterations} as a comparison metric. One outer iteration of Algorithm~\ref{alg:main_algo} corresponds to one round of communication between the center and the worker machines (and hence one parameter update). Note that in our algorithm the worker machines use 10 steps of gradient solver (see \cite{stoch}) for the local sub problem per iteration. So, the \emph{total number of iterations} is given by $10$ times the number of outer iterations.For both the algorithms, we choose $\ell_2$ norm of the gradient as a stopping criteria. For ByzantinePGD, we choose  $R=10,r=5,Q=10,T_{th}=10$ and `co-ordinate wise Trimmed mean. In the Figure~\ref{fig:saddle}, we plot the \emph{total  number of iterations} in all four types of attacks with different fraction of Byzantine machines. It is evident from the plot that our method requires less number of over all iterations (at least $48.4\%$, $29\%$ and $25\%$  less for $10\%$, $15\%$  and $20\%$ of Byzantine machines respectively). We provide the  results for the non-Byzantine case in the Appendix.

%% file: appendix5.tex
\newpage
\section{Appendix}\label{sec:app}

In this part, we establish some useful facts and lemmas. Next,
we provide  analysis of  Theorems~\ref{thm:oneround}. 

\subsection{Some useful facts}
For the purpose of analysis we use the following sets of inequalities.

\begin{fact}
 For $a_1,\ldots ,a_n$ we have the following inequality 
    \begin{align}
      \| \left( \sum_{i=1}^n a_i \right)\|^3 & \leq \left( \sum_{i=1}^n \|a_i\| \right)^3 \leq n^2\sum_{i=1}^n \|a_i \|^3 \label{eq:l3}\\
     \| \left( \sum_{i=1}^n a_i \right)\|^2  & \leq \left( \sum_{i=1}^n \|a_i\| \right)^2 \leq n\sum_{i=1}^n \|a_i \|^2 \label{eq:l2}
    \end{align}
\end{fact}

\begin{fact}
 For $a_1,\ldots,a_n>0 $ and $r<s$
    \begin{align}
        \left(\frac{1}{n}\sum_{i=1}^n a_i^r \right)^{1/r}\leq \left(\frac{1}{n}\sum_{i=1}^n a_i^s \right)^{1/s} \label{eq:pmean}
    \end{align}
\end{fact}

\begin{lemma}[\cite{nest}]\label{lem:Nest}
Under  Assumption ~\ref{asm:lip}, i.e., the Hessian of the function is $L_2$-Lipschitz continuous, for any $\x,\y \in \mathbb{R}^d$, we have 
\begin{align}
& \| \nabla f(\x) -\nabla f(\y) - \nabla^2 f(\x)(\y-\x)\| \leq \frac{L_2}{2}\|\y-\x\|^2 \label{eq:Nest1} \\
& \left| f(\y)-f(\x) - \nabla f(\x)^T(\y-\x) -\frac{1}{2}(\y-\x)^T\nabla^2 f(\x)(\y-\x) \right| \leq \frac{L_2}{6}\|\y-\x\|^2 \label{eq:Nest2}
\end{align}
\end{lemma} 
Next, we establish the following Lemma that provides some nice properties of the cubic sub-problem.
\begin{lemma}\label{lem:cubic}
Let $M>0,\gamma>0,\g \in \mathbb{R}^d,\bH \in \mathbb{R}^{d \times d}$, and 
\begin{align}
\s = \arg\min_{\x} \g^T\x + \frac{\gamma}{2}\x^T\bH \x + \frac{M\gamma^2}{6}\|\x\|^3.
\label{eq:cub}
\end{align}
The following holds
\begin{align}
\g + \gamma \bH\s + \frac{M\gamma^2}{2}\|\s\|\s & =0, \label{eq:first}\\
 \bH + \frac{M\gamma}{2}\|\s\|\bI & \succeq \mathbf{0}, \label{eq:second} \\
  \g^T\s + \frac{\gamma}{2}\s^T\bH\s & \leq - \frac{M}{4}\gamma^2 \|\s\|^3.  \label{eq:third}
\end{align}
\end{lemma}
\begin{proof}
The equations ~\eqref{eq:first} and ~\eqref{eq:second} are from the first and second order optimal condition. We proof ~\eqref{eq:third}, by using the conditions of ~\eqref{eq:first} and ~\eqref{eq:second}.
\begin{align}
    \g^T\s + \frac{\gamma}{2}\gamma\s^T\bH\s & = -\left( \gamma\bH\s + \frac{M}{2}\gamma^2\|\s\|\s\right)^T\s + \frac{\gamma}{2}\gamma\s^T\bH\s \label{eq:L1} \\
    & = - \gamma\s^T\bH\s - \frac{M}{2}\gamma^2\|\s\|^3 +\frac{\gamma}{2}\gamma\s^T\bH\s \nonumber \\
    & \leq   \frac{M}{4}\gamma^2\|\s\|^3- \frac{M}{2}\gamma^2\|\s\|^3\label{eq:L3} \\
    & =- \frac{M}{4}\gamma^2\|\s\|^3 \nonumber.
\end{align}
In ~\eqref{eq:L1}, we substitute the expression $\g$ from the equation~\eqref{eq:first}. In ~\eqref{eq:L3}, we use the fact that $\s^T\bH\s + \frac{M\gamma}{2}\|\s\|^3 >0$ from the equation ~\eqref{eq:second}.
\end{proof}

\subsection{ Proof of Theorem ~\ref{thm:oneround}}

First we state the results of Lemma~\ref{lem:cubic} for each worker machine in iteration $k$,
\begin{align}
\g_{i,k} + \gamma\bH_{i,k}\s_{i,k+1} + \frac{M}{2}\gamma^2\|\s_{i,k+1}\|\s_{i,k+1} & =0 \label{eq:T_first}\\
\gamma\bH_{i,k} + \frac{M}{2}\gamma^2\|\s_{i,k+1}\|\mathbf{I} & \succeq \mathbf{0} \label{eq:T_second} \\
\g_{i,k}^T\s_{i,k+1} + \frac{\gamma}{2}\s_{i,k+1}^T\bH_{i,k}\s_{i,k+1} & \leq - \frac{M}{4}\gamma^2 \|\s_{i,k+1}\|^3 \label{eq:T_third}
\end{align}
We also use the following fact form the setup and trimming set
\begin{align}
 | \cU | & = |\cU \cap \cM| + |\cU \cap \cB| \label{eq: B1}\\
 | \cM | & = |\cU \cap \cM| + |\cT \cap \cM| \label{eq:B2}   
\end{align} 
Combining both the equations ~\eqref{eq: B1} and ~\eqref{eq:B2}, we have
\begin{align}
| \cU | & = |\cM| -|\cT \cap \cM| + |\cU \cap \cB| \label{eq:B3}
\end{align}

Now we state the following fact from the trimming set 
\begin{align}
    \sum_{i \in \mathcal{U}\cap\mathcal{B}} \| .\| & \leq \alpha m \max_{i \in \mathcal{M}} \| .\|    \nonumber \\
    \sum_{i \in \mathcal{M}\cap\mathcal{T}}\|.\| &\leq \sum_{i \in \mathcal{M}} \|.\| \nonumber \\
    \intertext{ Combining both we get}
    \sum_{i \in \mathcal{U}\cap\mathcal{B}} \| .\|+ \sum_{i \in \mathcal{M}\cap\mathcal{T}}\|.\| & \leq \sum_{i \in \mathcal{M}} \|.\| + \alpha m \max_{i \in \mathcal{M}}\|.\| \label{eq:upbound}
\end{align}
For the rest of the calculation, we use the following notation 
\begin{align*}
    \Gamma = \max_{i \in \mathcal{M},k} \|\s_{i,k}\|.
\end{align*}
If the optimization sub-problem domain is bounded, $\Gamma$ can be upper-bounded by the diameter of the parameter space. Note that in the definition of $\Gamma$, the maximum is taken over good machines only.

From the definition of the $\delta$-approximate compressor in Definition~\ref{def:compress}, we use the following simple fact
\begin{align}
    \|Q(\x)\|\leq (1 + \sqrt{1-\delta})\|\x\| \label{eq:comp}
\end{align}

At any iteration $k$, we have 
\begin{align}
   & f(\x_{k+1})-f(\x_k) \nonumber \\
\leq & \nabla f(\x_k)^T(\x_{k+1}-\x_k) + \frac{1}{2}(\x_{k+1}-\x_k)^T \nabla^2 f(\x_k) (\x_{k+1}-\x_k) + \frac{L_2}{6}\left\|\x_{k+1}-\x_k \right\|^3  \nonumber  \\ 
= &\underbrace{ \frac{\eta_k}{|\mathcal{U}|}\nabla f(\x_k)^T\sum_{i \in \mathcal{U}}Q(\s_{i,k+1}) }_{Term 1}+ \underbrace{ \frac{\eta_k^2}{2|\mathcal{U}|^2}\left(\sum_{i \in \mathcal{U}} Q(\s_{i,k+1}) \right)^T \nabla^2 f(\x_k)\left(\sum_{i \in \mathcal{U}} Q(\s_{i,k+1}) \right) }_{Term 2} \nonumber  \\
& + \underbrace{ \frac{L_2}{6} \left\|\frac{\eta_k}{|\mathcal{U}|}\sum_{i \in \mathcal{U}}Q(\s_{i,k+1}) \right\|^3 }_{Term 3} \label{eq:T31}
\end{align}

First we choose the Term 1 in \eqref{eq:T31} and expand it using ~\eqref{eq:B3} 
\begin{align}
    & \frac{\eta_k}{|\mathcal{U}|}\nabla f(\x_k)^T\sum_{i \in \mathcal{U}} Q(\s_{i,k+1}) \nonumber  \\
= & \frac{\eta_k}{(1-\beta)m}\nabla f(\x_k)^T\left[ \sum_{i \in \mathcal{M}}Q(\s_{i,k+1})- \sum_{i \in \mathcal{M}\cap\mathcal{T}}Q(\s_{i,k+1})  +\sum_{i \in \mathcal{U}\cap\mathcal{B}}Q(\s_{i,k+1}) \right] \nonumber \\ 
 = & \frac{\eta_k}{(1-\beta)m}\sum_{i \in \mathcal{M}}\g_{i,k}^T\s_{i,k+1}+ \frac{\eta_k}{(1-\beta)m} \sum_{i \in \mathcal{M}}\left[\nabla f(\x_k)^T Q(\s_{i,k+1})- \g_{i,k}^T\s_{i,k+1} \right] \nonumber \\
& + \frac{\eta_k}{(1-\beta)m} \nabla f(\x_k)^T\left[- \sum_{i \in \mathcal{M}\cap\mathcal{T}}Q(\s_{i,k+1})  +\sum_{i \in \mathcal{U}\cap\mathcal{B}}Q(\s_{i,k+1}) \right] \nonumber  \\
 \leq & \frac{\eta_k}{(1-\beta)m}\sum_{i \in \mathcal{M}}\g_{i,k}^T\s_{i,k+1}+ \frac{\eta_k}{(1-\beta)m} \sum_{i \in \mathcal{M}}\left[\nabla f(\x_k)^T Q(\s_{i,k+1}) -\nabla f(\x_k)^T \s_{i,k+1}+\nabla f(\x_k)^T \s_{i,k+1} - \g_{i,k}^T\s_{i,k+1} \right] \nonumber \\
& + \frac{\eta_k}{(1-\beta)m} \left[ \sum_{i \in \mathcal{M}\cap\mathcal{T}}
\|\nabla f(\x_k)\| \|Q(\s_{i,k+1})\|  +\sum_{i \in \mathcal{U}\cap\mathcal{B}} \|\nabla f(\x_k)\|\|Q(\s_{i,k+1})\| \right] \nonumber \\
\leq & \frac{\eta_k}{(1-\beta)m}\sum_{i \in \mathcal{M}}\g_{i,k}^T\s_{i,k+1}+ \frac{\eta_k}{(1-\beta)m} \sum_{i \in \mathcal{M}}\left[\| \nabla f(\x_k)\|\| Q(\s_{i,k+1}) -\s_{i,k+1}\| + \|\nabla f(\x_k) - \g_{i,k}\| \|\s_{i,k+1}\| \right] \nonumber \\
& + \frac{\eta_k L}{(1-\beta)m} \left[ \sum_{i \in \mathcal{M}\cap\mathcal{T}}
 \|Q(\s_{i,k+1})\|  +\sum_{i \in \mathcal{U}\cap\mathcal{B}} \|Q(\s_{i,k+1})\| \right] \nonumber \\
 \leq & \frac{\eta_k}{(1-\beta)m}\sum_{i \in \mathcal{M}}\g_{i,k}^T\s_{i,k+1}+ \frac{\eta_k}{(1-\beta)m} \sum_{i \in \mathcal{M}}\left[L\sqrt{1-\delta} \| \s_{i,k+1}\| + \epsilon_g \|\s_{i,k+1}\| \right] \nonumber \\
& + \frac{\eta_k L}{(1-\beta)m} \left[ \sum_{i \in \mathcal{M}\cap\mathcal{T}}
 \|Q(\s_{i,k+1})\|  +\sum_{i \in \mathcal{U}\cap\mathcal{B}} \|Q(\s_{i,k+1})\| \right] \nonumber
\end{align}

In the above calculation we use the fact gradients are bounded $\| \nabla f(\x_k)\| \leq L $ and the bound of gradient dissimilarity. 

We use the result ~\eqref{eq:upbound}, we get
\begin{align}
   & Term 1 \nonumber \\
   \leq & \frac{\eta_k}{(1-\beta)m}\sum_{i \in \mathcal{M}}\g_{i,k}^T\s_{i,k+1}+ \frac{\eta_k}{(1-\beta)m} \sum_{i \in \mathcal{M}}\left[L\sqrt{1-\delta} \| \s_{i,k+1}\| + \epsilon_g \|\s_{i,k+1}\| \right] \nonumber \\
& + \frac{\eta_k L (1 +\sqrt{1-\delta})}{(1-\beta)m} \left[ \sum_{i \in \mathcal{M}} \|Q(\s_{i,k+1})\| + \alpha m \Gamma\right] \nonumber\\
 \leq & \frac{\eta_k}{(1-\beta)m}\sum_{i \in \mathcal{M}}\g_{i,k}^T\s_{i,k+1}+ \frac{\eta_k(1-\alpha)}{(1-\beta)} (L\sqrt{1-\delta} + \epsilon_g) \Gamma   + \frac{\eta_k L}{(1-\beta)}  (1+ \sqrt{1-\delta})\Gamma \nonumber \\
 = & \frac{\eta_k}{(1-\beta)m}\sum_{i \in \mathcal{M}}\left[ \g_{i,k}^T\s_{i,k+1}+ \frac{\gamma}{2}\s_{i,k+1}^T\bH_{i,k}\s_{i,k+1} \right]- \frac{\eta_k}{(1-\beta)m}\sum_{i \in \mathcal{M}}\frac{\gamma}{2}\s_{i,k+1}^T\bH_{i,k}\s_{i,k+1} \nonumber \\
 &+\frac{\eta_k(1-\alpha)}{(1-\beta)} (L\sqrt{1-\delta} + \epsilon_g) \Gamma   + \frac{\eta_k L}{(1-\beta)}  (1+ \sqrt{1-\delta})\Gamma \nonumber \\
\leq  & -\frac{\gamma^2M\eta_k}{4(1-\beta)m}\sum_{i \in \mathcal{M}}\|\s_{i,k+1}\|^3- \frac{\eta_k}{(1-\beta)m}\sum_{i \in \mathcal{M}}\frac{\gamma}{2}\s_{i,k+1}^T\bH_{i,k}\s_{i,k+1} \nonumber \\
 & +\frac{\eta_k(1-\alpha)}{(1-\beta)} (L\sqrt{1-\delta} + \epsilon_g) \Gamma   + \frac{\eta_k L}{(1-\beta)}  (1+ \sqrt{1-\delta})\Gamma \label{eq:T311}
\end{align}

In line ~\eqref{eq:T311}, we use the bound stated in ~\eqref{eq:T_third}.
Now we consider the Term 3 in equation ~\eqref{eq:T31},

\begin{align}
    &\frac{L_2}{6} \left\|\frac{\eta_k}{|\mathcal{U}|}\sum_{i \in \mathcal{U}}Q(\s_{i,k+1}) \right\|^3 \nonumber \\
    \leq &  \frac{L_2\eta_k^3}{6|\mathcal{U}|}\sum_{i \in \mathcal{U}} \left\|Q(\s_{i,k+1}) \right\|^3 \nonumber \\
    \leq &   \frac{L_2\eta_k^3}{6|\mathcal{U}|} \left[\sum_{i \in \mathcal{M}}\|Q(\s_{i,k+1})\|^3- \sum_{i \in \mathcal{M}\cap\mathcal{T}}\|Q(\s_{i,k+1}) \|^3+\sum_{i \in \mathcal{U}\cap\mathcal{B}}\|Q(\s_{i,k+1})\|^3 \right] \nonumber\\
    \leq &   \frac{L_2\eta_k^3}{6|\mathcal{U}|} \left[\sum_{i \in \mathcal{M}}\|Q(\s_{i,k+1})\|^3+\sum_{i \in \mathcal{U}\cap\mathcal{B}}\|Q(\s_{i,k+1})\|^3 \right] \nonumber
\end{align}
In the first inequality we use the fact stated in  ~\eqref{eq:l3}. Next we expand the trimmed set $\mathcal{U}$ using \eqref{eq:B3}. 

Now we use the fact in equation ~\eqref{eq:upbound} and the definition of the $\delta$-approximate compressor, we get
\begin{align}
    Term 3 \leq & \frac{L_2\eta_k^3}{6(1-\beta)m} \left[\sum_{i \in \mathcal{M}} (1+\sqrt{1-\delta})^3\|\s_{i,k+1}\|^3 + (1+\sqrt{1-\delta})^3(\alpha m)\Gamma^3  \right]  
\end{align}

Now we consider the Term 2 in ~\eqref{eq:T31} 
\begin{align}
   & \frac{\eta_k^2}{2|\mathcal{U}|^2}\left(\sum_{i \in \mathcal{U}} Q(\s_{i,k+1}) \right)^T \nabla^2 f(\x_k)\left(\sum_{i \in \mathcal{U}} Q(\s_{i,k+1}) \right) \nonumber \\
   & = \frac{\eta_k^2}{2(1-\beta)^2m^2}\left( \underbrace{\sum_{i \in \mathcal{U}}Q(\s_{i,k+1})^T \nabla^2 f(\x_k)Q(\s_{i,k+1})}_{Term 4}+ \underbrace{\sum_{i\neq j \in \mathcal{U}}Q(\s_{i,k+1})^T \nabla^2 f(\x_k)Q(\s_{j,k+1}) }_{Term 5}\right)\label{eq:Te2}
\end{align}

Now we consider Term 4 in \eqref{eq:Te2} and expand it using ~\eqref{eq:B3}
\begin{align}
& \sum_{i \in \mathcal{U}}Q(\s_{i,k+1})^T \nabla^2 f(\x_k)Q(\s_{i,k+1}) \nonumber \\
=&  \sum_{i \in \mathcal{M}}Q(\s_{i,k+1})^T \nabla^2 f(\x_k)Q(\s_{i,k+1}) -\sum_{i \in \mathcal{M}\cap \mathcal{T} }Q(\s_{i,k+1})^T \nabla^2 f(\x_k)Q(\s_{i,k+1}) \nonumber \\& + \sum_{i \in \mathcal{B}\cap \mathcal{U} }Q(\s_{i,k+1})^T \nabla^2 f(\x_k)Q(\s_{i,k+1}) \nonumber  \\
\leq  & \sum_{i \in \mathcal{M}}\s_{i,k+1}^T \bH_{i,k}\s_{i,k+1} - \sum_{i \in \mathcal{M}}\s_{i,k+1}^T \bH_{i,k}\s_{i,k+1} +\sum_{i \in \mathcal{M}}Q(\s_{i,k+1})^T \nabla^2 f(\x_k)Q(\s_{i,k+1})  \nonumber \\
& \qquad  + \sum_{i \in \mathcal{M}\cap \mathcal{T} }L_1 \|Q(\s_{i,k+1})\|^2 + \sum_{i \in \mathcal{B}\cap \mathcal{U} } L_1 \|Q(\s_{i,k+1}) \|^2 \nonumber \\ 
\leq  & \sum_{i \in \mathcal{M}}\s_{i,k+1}^T \bH_{i,k}\s_{i,k+1} - \sum_{i \in \mathcal{M}}\s_{i,k+1}^T (\nabla f(\x_k)-\bH_{i,k})\s_{i,k+1} + \sum_{i \in \mathcal{M}}\s_{i,k+1}^T \nabla f(\x_k)\s_{i,k+1}\nonumber \\
& \qquad +  \sum_{i \in \mathcal{M}}Q(\s_{i,k+1})^T \nabla^2 f(\x_k)Q(\s_{i,k+1})   + \sum_{i \in \mathcal{M}\cap \mathcal{T} }L_1 \|Q(\s_{i,k+1})\|^2 + \sum_{i \in \mathcal{B}\cap \mathcal{U} } L_1 \|Q(\s_{i,k+1}) \|^2 \nonumber \\ 
\leq &  \sum_{i \in \mathcal{M}}\s_{i,k+1}^T \bH_{i,k}\s_{i,k+1} + \sum_{i \in \mathcal{M}}(\epsilon_H+L_1) \|\s_{i,k+1}\|^2 \nonumber \\
\qquad &+   \sum_{i \in \mathcal{M}}L_1 \| Q(\s_{i,k+1})\|^2   + \sum_{i \in \mathcal{M} } L_1  \|Q(\s_{i,k+1}) \|^2 + L_1\alpha m (1+\sqrt{1-\delta})^2\Gamma^2 \nonumber \\
\leq &  \sum_{i \in \mathcal{M}}\s_{i,k+1}^T \bH_{i,k}\s_{i,k+1} + \sum_{i \in \mathcal{M}}(\epsilon_H+L_1) \|\s_{i,k+1}\|^2 + \sum_{i \in \mathcal{M}}L_1 2(1+\sqrt{1-\delta})^2 \|\s_{i,k+1}\|^2 + L_1\alpha m (1+\sqrt{1-\delta})^2\Gamma^2 \nonumber  \\
& = \sum_{i \in \mathcal{M}}\s_{i,k+1})^T \bH_{i,k}\s_{i,k+1} + \sum_{i \in \mathcal{M}}[\epsilon_H+L_1(1+ 2(1+ \sqrt{1-\delta})^2)  ]\|\s_{i,k+1}\|^2 + L_1\alpha m(1+\sqrt{1-\delta})^2\Gamma^2 \label{eq:Te3}
\end{align}

Now we consider the  Term 5 in equation ~\eqref{eq:Te3}
\begin{align}
   & \sum_{i\neq j \in \mathcal{U}}Q(\s_{i,k+1})^T \nabla^2 f(\x_k)Q(\s_{j,k+1})  \nonumber \\
  \leq  & \sum_{i\neq j \in \mathcal{U}}L_1\|Q(\s_{i,k+1})\|\|Q(\s_{j,k+1})\|  \nonumber \\ 
  \leq &  \sum_{i\neq j \in \mathcal{U}}L_1(1+ \sqrt{1-\delta})^2\|\s_{i,k+1}\|\|\s_{j,k+1}\|  \nonumber \\
  = & L_1(1+ \sqrt{1-\delta})^2\left[\|\sum_{i \in \mathcal{U}}\s_{i,k+1}\|^2- \sum_{i \in \mathcal{U}}\|\s_{i,k+1}\|^2\right] \nonumber \\
  \leq  & L_1(1+ \sqrt{1-\delta})^2\left[|\mathcal{U}|\sum_{i \in \mathcal{U}}\|\s_{i,k+1}\|^2- \sum_{i \in \mathcal{U}}\|\s_{i,k+1}\|^2\right] \nonumber \\
  = & L_1(1+ \sqrt{1-\delta})^2( (1-\beta)m-1)\left[ \sum_{i \in \mathcal{M}}\|\s_{i,k+1}\|^2 - \sum_{i \in \mathcal{M}\cap \mathcal{T} }\|\s_{i,k+1}\|^2 + \sum_{i \in \mathcal{B}\cap \mathcal{U} }\|\s_{i,k+1}\|^2\right] \nonumber \\
  \leq & L_1(1+ \sqrt{1-\delta})^2( (1-\beta)m-1)\left[ \sum_{i \in \mathcal{M}}\|\s_{i,k+1}\|^2 + \alpha m\Gamma^2\right]  \label{eq:Te4}
\end{align}
Now combining the results in \eqref{eq:Te3} and \eqref{eq:Te2} we get,

\begin{align}
   & Term 2 \nonumber \\ \leq &  \frac{\eta_k^2}{2(1-\beta)^2m^2}\sum_{i \in \mathcal{M}}\s_{i,k+1})^T \bH_{i,k}\s_{i,k+1}+ \frac{\eta_k^2(1-\alpha)}{2(1-\beta)^2m}\left[\epsilon_H+L_1(1+(1-\beta)m+ (1+ \sqrt{1-\delta})^2)   \right]\Gamma^2 \nonumber \\
   & +\frac{\eta_k^2}{2(1-\beta)^2m^2}L_1(1+ \sqrt{1-\delta})^2 \alpha(1-\beta)m^2\Gamma^2
\end{align}

Now we combine all the upper bound of the Term 1, Term 2 and Term 3 

\begin{align*}
     & f(\x_{k+1})-f(\x_k) \\
    \leq   & [-\frac{\gamma^2M\eta_k}{4(1-\beta)m} +\frac{L_2\eta_k^3}{6(1-\beta)m} (1+\sqrt{1-\delta})^3  ]\sum_{i \in \mathcal{M}}\|\s_{i,k+1}\|^3  - \frac{\eta_k}{2(1-\beta)m} [ \gamma-  \frac{\eta_k}{(1-\beta)m} ] \sum_{i \in \mathcal{M}}\s_{i,k+1}^T\bH_{i,k}\s_{i,k+1} \\ &+\frac{\eta_k(1-\alpha)}{(1-\beta)} (L\sqrt{1-\delta} + \epsilon_g) \Gamma   + \frac{\eta_k L}{(1-\beta)}  (1+ \sqrt{1-\delta})\Gamma\\ & + \frac{\eta_k^2(1-\alpha)}{2(1-\beta)^2m}\left[\epsilon_H+L_1(1+(1-\beta)m+ (1+ \sqrt{1-\delta})^2)   \right]\Gamma^2 \nonumber \\
   & +\frac{\eta_k^2}{2(1-\beta)^2}L_1(1+ \sqrt{1-\delta})^2 \alpha(1-\beta)\Gamma^2  + \frac{L_2\eta_k^3}{6(1-\beta)} (1+\sqrt{1-\delta})^3\alpha \Gamma^3  
\end{align*}
 Also we assume that $\gamma \geq \frac{\eta_k}{(1-\beta)m}$ and use the fact  $-\s_{i,k+1}\bH_{i,k}\s_{i,k+1}\leq \frac{M\gamma}{2}\|\s_{i,k+1}\|^3 $.

Now we have, 
\begin{align*}
    & f(\x_{k+1})-f(\x_k) \\
\leq   & [-\frac{\gamma^2M}{4(1-\beta)\eta_k^2m} +\frac{L_2}{6(1-\beta)m} (1+\sqrt{1-\delta})^3  ]\sum_{i \in \mathcal{M}}\|\eta_k\s_{i,k+1}\|^3  - \frac{\eta_k}{2(1-\beta)m} [ \gamma-  \frac{\eta_k}{(1-\beta)m} ] \sum_{i \in \mathcal{M}}\frac{M\gamma}{2}\|\s_{i,k+1}\|^3 \\ 
&+\frac{\eta_k(1-\alpha)}{(1-\beta)} (L\sqrt{1-\delta} + \epsilon_g) \Gamma   + \frac{\eta_k L}{(1-\beta)}  (1+ \sqrt{1-\delta})\Gamma \\ & + \frac{\eta_k^2(1-\alpha)}{2(1-\beta)^2m}\left[\epsilon_H+L_1(1+(1-\beta)m+ (1+ \sqrt{1-\delta})^2)   \right]\Gamma^2 \nonumber \\
   & +\frac{\eta_k^2}{2(1-\beta)^2}L_1(1+ \sqrt{1-\delta})^2 \alpha(1-\beta)\Gamma^2  + \frac{L_2\eta_k^3}{6(1-\beta)} (1+\sqrt{1-\delta})^3\alpha \Gamma^3 \\
 \leq & [-\frac{\gamma M}{4(1-\beta)\eta_k m^2} +\frac{L_2}{6(1-\beta)m} (1+\sqrt{1-\delta})^3  ]\sum_{i \in \mathcal{M}}\|\eta_k\s_{i,k+1}\|^3  \\
 &+\frac{\eta_k(1-\alpha)}{(1-\beta)} (L\sqrt{1-\delta} + \epsilon_g) \Gamma   + \frac{\eta_k L}{(1-\beta)}  (1+ \sqrt{1-\delta})\Gamma \\ &+ \frac{\eta_k^2(1-\alpha)}{2(1-\beta)^2m}\left[\epsilon_H+L_1(1+(1-\beta)m+ (1+ \sqrt{1-\delta})^2)   \right]\Gamma^2 \nonumber \\
   & +\frac{\eta_k^2}{2(1-\beta)^2}L_1(1+ \sqrt{1-\delta})^2 \alpha(1-\beta)\Gamma^2  + \frac{L_2\eta_k^3}{6(1-\beta)} (1+\sqrt{1-\delta})^3\alpha \Gamma^3 \\
  = & -\lambda_{comp}\frac{1}{(1-\alpha)m} \sum_{i \in \mathcal{M}}  \|\eta_k\s_{i,k+1}\|^3 + \lambda_{\Gamma}
\end{align*}

where 
\begin{align*}
    \lambda_{comp} =[\frac{\gamma M}{4(1-\beta)\eta_k m^2} -\frac{L_2}{6(1-\beta)m} (1+\sqrt{1-\delta})^3  ](1-\alpha)m
\end{align*}

and
\begin{align*}
    \lambda_{\Gamma}& = \frac{\eta_k(1-\alpha)}{(1-\beta)} \left[(1-\alpha)(L\sqrt{1-\delta} + \epsilon_g) +   L  (1+ \sqrt{1-\delta})\right]\Gamma  + + \frac{L_2\eta_k^3}{6(1-\beta)} (1+\sqrt{1-\delta})^3\alpha \Gamma^3 \\
    & + \frac{\eta_k^2(1-\alpha)}{2(1-\beta)^2m}\left[\epsilon_H+L_1(1+(1-\beta)m+ (1+ \sqrt{1-\delta})^2)   \right]\Gamma^2 
    +\frac{\eta_k^2}{2(1-\beta)^2}L_1(1+ \sqrt{1-\delta})^2 \alpha(1-\beta)\Gamma^2  
\end{align*}

To ensure $\lambda_{comp}>0$, we need
\begin{align}
    M > \frac{4\eta_k^2m}{\gamma^2}\frac{L_2}{6} (1+\sqrt{1-\delta})^3 \label{eq:Mvalue}
\end{align}

Now for the choice of $\eta_k =\frac{c}{T}$ and $\gamma = \frac{c_1}{T}$ for some constant $c,c_1 >0$. We have $M = \mathcal{O}(m)$. Thus we have 
$\lambda_{comp}= \mathcal{O}(1) $ and $\lambda_{\Gamma}= \mathcal{O}(\frac1T) + \text{minor terms}$.

Now we have 
\begin{align*}
\frac{1}{(1-\alpha)m}\sum_{i \in \mathcal{M}}\|\eta_k\s_{i,k+1}\|^3 \leq \frac{1}{\lambda_{comp}}\left[f(\x_k) -f(\x_{k+1})   + \lambda_{\Gamma}\right]
\end{align*}

Now we consider the step $k_0$, where  $k_0= \arg\min_{0\leq k \leq T-1}\|\x_{k+1}-\x_{k}\|=\arg\min_{0\leq k \leq T-1}\|\eta_kQ(\s_{k+1})\| $. 
\begin{align*}
  \min_{0\leq k \leq T}\|\x_{k+1}-\x_{k}\|^3 &=  \min_{0\leq k \leq T}\|\eta_k\ Q(s_{k+1})\|^3 \\
  & \leq \frac{1}{(1-\beta)m}\sum_{i\in \mathcal{U}}\|\eta_{k_0}Q(\s_{i,k_0+1})\|^3 \\
  & \leq \frac{(1+\sqrt{1-\delta})^3}{(1-\beta)m}\sum_{i\in \mathcal{U}}\|\eta_{k_0}\s_{i,k_0+1}\|^3 \\
  & = \frac{(1+\sqrt{1-\delta})^3}{(1-\beta)m} \left[ \sum_{i\in \mathcal{M}}\|\eta_{k_0}\s_{i,k_0+1}\|^3 - \sum_{i\in \mathcal{M}\cap \mathcal{T}}\|\eta_{k_0}\s_{i,k_0+1}\|^3 +\sum_{i\in \mathcal{U}\cap \mathcal{B}}\|\eta_{k_0}\s_{i,k_0+1}\|^3\right] \\
  & \leq \frac{(1+\sqrt{1-\delta})^3}{(1-\beta)m} \left[ \sum_{i\in \mathcal{M}}\|\eta_{k_0}\s_{i,k_0+1}\|^3 + \alpha m \eta^3_{k_0}\Gamma^3\right] \\
   & \leq \frac{1}{T}\sum_{k=0}^{T-1}(1+\sqrt{1-\delta})^3\frac{(1-\alpha)}{(1-\beta)}\left[\frac{1}{(1-\alpha)m} \sum_{i\in \mathcal{M}}\|\eta_{k}\s_{i,k+1}\|^3 + \frac{\alpha}{1-\alpha}\eta^3_{k_0}\Gamma^3 \right] \\
   & \leq \frac{1}{T}\sum_{k=0}^{T-1}(1+\sqrt{1-\delta})^3\frac{(1-\alpha)}{(1-\beta)}\left[ \frac{f(\x_{k}) -f(\x_{k+1})}{\lambda_{comp}}+  \frac{\lambda_{\Gamma}}{\lambda_{comp}} + \frac{\alpha}{1-\alpha}\frac{\eta^3_{k_0}\Gamma^3}{\lambda_{comp}} \right]\\
   & \leq \frac{1}{T}(1+\sqrt{1-\delta})^3\frac{(1-\alpha)}{(1-\beta)}\left[ \frac{f(\x_{0}) -f^*}{\lambda_{comp}}+ \sum_{k=0}^{T-1}\frac{\lambda_{\Gamma}}{\lambda_{comp}} + \sum_{k=0}^{T-1}\frac{\alpha}{1-\alpha}\frac{\eta^3_{k_0}\Gamma^3}{\lambda_{comp}} \right]
\end{align*}

With the choice of  $\eta_k,\gamma$ we have the terms 
$\sum_{k=0}^{T-1}\frac{\lambda_{\Gamma}}{\lambda_{comp}}$  and $\sum_{k=0}^{T-1}\frac{\alpha}{1-\alpha}\frac{\eta^3_{k_0}\Gamma^3}{\lambda_{comp}}$ are upper bounded by constant.

We have 
\begin{align*}
    \frac{1}{(1-\beta)m} \left[ \sum_{i\in \mathcal{M}}\|\eta_{k_0}\s_{i,k_0+1}\|^3 + \alpha m \eta^3_{k_0}\Gamma^3\right] & \leq \frac{1}{T}\frac{(1-\alpha)}{(1-\beta)}\left[ \frac{f(\x_{0}) -f^*}{\lambda_{comp}}+ \sum_{k=0}^{T-1}\frac{\lambda_{\Gamma}}{\lambda_{comp}} + \sum_{k=0}^{T-1}\frac{\alpha}{1-\alpha}\frac{\eta^3_{k_0}\Gamma^3}{\lambda_{comp}} \right] \\
    \Rightarrow \frac{1}{(1-\alpha)m}\left[ \sum_{i\in \mathcal{M}}\|\eta_{k_0}\s_{i,k_0+1}\|^3 + \alpha m \eta^3_{k_0}\Gamma^3\right] & \leq \frac{1}{T}\left[ \frac{f(\x_{0}) -f^*}{\lambda_{comp}}+ \sum_{k=0}^{T-1}\frac{\lambda_{\Gamma}}{\lambda_{comp}} + \sum_{k=0}^{T-1}\frac{\alpha}{1-\alpha}\frac{\eta^3_{k_0}\Gamma^3}{\lambda_{comp}} \right]\\
    \Rightarrow \frac{1}{(1-\alpha)m} \sum_{i\in \mathcal{M}}\|\eta_{k_0}\s_{i,k_0+1}\|^3  & \leq \frac{1}{T}\left[ \frac{f(\x_{0}) -f^*}{\lambda_{comp}}+ \sum_{k=0}^{T-1}\frac{\lambda_{\Gamma}}{\lambda_{comp}}  \right] =\frac{\psi_{comp}}{T}
\end{align*}
where $\psi_{comp}=\left[ \frac{f(\x_{0}) -f^*}{\lambda_{comp}}+ \sum_{k=0}^{T-1}\frac{\lambda_{\Gamma}}{\lambda_{comp}}  \right] $

And
\begin{align*}
  \|\x_{k_0+1}-\x_{k_0}\| \leq \frac{\psi_{comp,1}}{T^{1/3}}  
\end{align*}
where 
$\psi_{comp,1}= (1+\sqrt{1-\delta})\frac{(1-\alpha)^{1/3}}{(1-\beta)^{1/3}}\left[ \frac{f(\x_{0}) -f^*}{\lambda_{comp}}+ \sum_{k=0}^{T-1}\frac{\lambda_{\Gamma}}{\lambda_{comp}} + \sum_{k=0}^{T-1}\frac{\alpha}{1-\alpha}\frac{\eta^3_{k_0}\Gamma^3}{\lambda_{comp}} \right]^{1/3}$

So, we have both the term $\psi_{comp},\psi_{comp,1}$ are of the order $\mathcal{O}(1)$.

The gradient condition is 

\begin{align}
  &\left\| \nabla f(\x_{k+1}) \right\| \nonumber \\  & =  \left\|\nabla f(\x_{k+1})- \frac{1}{|\mathcal{M}|}\sum_{i \in \mathcal{M}} \g_{i,k}-\frac{1}{|\mathcal{M}|}\sum_{i \in \mathcal{M}} \gamma\bH_{i,k+1}\s_{i,k+1}- \frac{1}{|\mathcal{M}|}\sum_{i \in \mathcal{M}}\frac{M\gamma^2}{2}\|\s_{i,k+1}\|\s_{i,k+1} \right\|  \nonumber \\
    & \leq  \left\|\nabla f(\x_{k+1})- \nabla f(\x_k) - \nabla^2 f(\x_k)(x_{k+1}-x_k) \right\| + \left\|\frac{1}{|\mathcal{M}|}\sum_{i \in \mathcal{M}}(\g_{i,k}-\nabla f(\x_k) ) \right\|  \nonumber\\ & \qquad +  \left\|  \nabla^2 f(\x_k)(x_{k+1}-x_k) -\gamma\frac{1}{|\mathcal{M}|}\sum_{i \in \mathcal{M}}\bH_{i,k}\s_{i,k+1} \right\| + \left\| \frac{1}{|\mathcal{M}|}\sum_{i \in \mathcal{M}} \frac{M\gamma^2}{2}\|\s_{i,k+1}\|\s_{i,k+1} \right\| \nonumber \\
     & \leq \frac{L_2\eta_k^2}{2}\left\|\frac{1}{|\mathcal{U}|}\sum_{i \in \mathcal{U}}Q(\s_{i,k+1} )\right\|^2 + \epsilon_g + \frac{M\gamma^2}{2}\frac{1}{|\mathcal{M}|}\sum_{i \in \mathcal{M}} \|\s_{i,k+1}\|^2  \nonumber \\ 
     & +  \left\| \frac{\eta_k}{|\mathcal{U}|} \sum_{i \in \mathcal{U}}\nabla^2 f(\x_k)Q(\s_{i,k+1}) - \frac{\gamma}{|\mathcal{M}|} \sum_{i \in \mathcal{M}} \bH_{i,k}\s_{i,k+1} \right\| \label{eq:T3g1}
\end{align}

Now consider the term in \eqref{eq:T3g1}
\begin{align}
   & \left\| \frac{\eta_k}{|\mathcal{U}|} \sum_{i \in \mathcal{U}}\nabla^2 f(\x_k)Q(\s_{i,k+1}) - \frac{\gamma}{|\mathcal{M}|} \sum_{i \in \mathcal{M}} \bH_{i,k}\s_{i,k+1} \right\| \nonumber \\
\leq & \left\| \frac{\eta_k}{|\mathcal{U}|} \left[\sum_{i \in \mathcal{M}}\nabla^2 f(\x_k)Q(\s_{i,k+1}) - \sum_{i \in \mathcal{M}\cap \mathcal{T}}\nabla^2 f(\x_k)Q(\s_{i,k+1}) + \sum_{i \in \mathcal{B}\cap \mathcal{U}}\nabla^2 f(\x_k)Q(\s_{i,k+1}) \right] - \frac{\gamma}{|\mathcal{M}|} \sum_{i \in \mathcal{M}} \bH_{i,k}\s_{i,k+1} \right\| \nonumber \\ 
\leq & \left\| \frac{\eta_k}{|\mathcal{U}|}\sum_{i \in \mathcal{M}}\nabla^2 f(\x_k)Q(\s_{i,k+1})   - \frac{\gamma}{|\mathcal{M}|} \sum_{i \in \mathcal{M}} \bH_{i,k}\s_{i,k+1} \right\| +  \frac{\eta_k}{|\mathcal{U}|}\sum_{i \in \mathcal{M}\cap \mathcal{T}}\|\nabla^2 f(\x_k)Q(\s_{i,k+1})\| \nonumber \\& +  \frac{\eta_k}{|\mathcal{U}|}\sum_{i \in \mathcal{B}\cap \mathcal{U}}\|\nabla^2 f(\x_k)Q(\s_{i,k+1})\| \nonumber \\
\leq & \left\| \frac{\eta_k}{|\mathcal{U}|}\sum_{i \in \mathcal{M}}\nabla^2 f(\x_k)Q(\s_{i,k+1})   - \frac{\gamma}{|\mathcal{M}|} \sum_{i \in \mathcal{M}} \nabla^2 f(\x_k)Q(\s_{i,k+1}) \right\| \nonumber \\ &+ \left\| \frac{\gamma}{|\mathcal{M}|}\sum_{i \in \mathcal{M}}\nabla^2 f(\x_k)Q(\s_{i,k+1})   - \frac{\gamma}{|\mathcal{M}|} \sum_{i \in \mathcal{M}} \nabla^2 f(\x_k)\s_{i,k+1} \right\| \nonumber \\
& + \left\| \frac{\gamma}{|\mathcal{M}|}\sum_{i \in \mathcal{M}}\nabla^2 f(\x_k)\s_{i,k+1}   - \frac{\gamma}{|\mathcal{M}|} \sum_{i \in \mathcal{M}} \bH_{i,k}\s_{i,k+1} \right\| + \frac{\eta_k}{(1-\beta)m}(1+ \sqrt{1-\delta})L_1 [\sum_{i \in \mathcal{M}}\|\s_{i,k+1}\| + \alpha m \Gamma] \nonumber \\
\leq & (\frac{\eta_k}{(1-\beta)m}- \frac{\gamma}{(1-\alpha)m})L_1(1+ \sqrt{1-\delta})\sum_{i \in \mathcal{M}}\|\s_{i,k+1}\| + \frac{\gamma}{(1-\alpha)m}L_1\sqrt{1-\delta} \sum_{i \in \mathcal{M}}\|\s_{i,k+1}\| \nonumber \\
& + \frac{\gamma \epsilon_H}{(1-\alpha)m}\sum_{i \in \mathcal{M}}\|\s_{i,k+1}\| + \frac{\eta_k}{(1-\beta)m}(1+ \sqrt{1-\delta})L_1[\sum_{i \in \mathcal{M}}\|\s_{i,k+1}\| + \alpha m \Gamma ]\nonumber \\
& \leq \left( \frac{\eta_k}{(1-\beta)m}L_1(1+ \sqrt{1-\delta})(2+ \alpha m) -\frac{\gamma L_1}{(1-\alpha)m} \right)\sum_{i \in \mathcal{M}}\|\s_{i,k+1}\| + \frac{\gamma \epsilon_H}{(1-\alpha)m}\sum_{i \in \mathcal{M}}\|\s_{i,k+1}\| \nonumber \\
&= \left( \frac{(1-\alpha)}{(1-\beta)}2L_1(1+ \sqrt{1-\delta}) -\frac{\gamma}{\eta_k} (L_1- \epsilon_H)\right)\frac{1}{(1-\alpha)m} \sum_{i \in \mathcal{M}}\|\eta_k\s_{i,k+1}\| + \frac{\eta_k\alpha}{(1-\beta)}(1 + \sqrt{1-\delta})\Gamma \nonumber 
\end{align}

Next we consider the term 

\begin{align}
    & \frac{L_2\eta_k^2}{2}\left\|\frac{1}{|\mathcal{U}|}\sum_{i \in \mathcal{U}}Q(\s_{i,k+1} )\right\|^2 \nonumber \\
 \leq &     \frac{L_2(1+\sqrt{1-\delta})^2\eta_k^2}{2(1-\beta)m}\sum_{i \in \mathcal{U}}\|\s_{i,k+1} \|^2 \nonumber \\
 \leq &     \frac{L_2(1+\sqrt{1-\delta})^2\eta_k^2}{2(1-\beta)m} \left[\sum_{i \in \mathcal{M}}\|\s_{i,k+1} \|^2 +\sum_{i \in \mathcal{U} \cap \mathcal{B}}\|\s_{i,k+1} \|^2 \right]\nonumber \\ 
 = & \frac{L_2(1+\sqrt{1-\delta})^2\eta_k^2}{2(1-\beta)m} \sum_{i \in \mathcal{M}}\|\s_{i,k+1} \|^2 + \frac{L_2\alpha(1+\sqrt{1-\delta})^2\eta_k^2}{2(1-\beta)}\Gamma^2
\end{align}

So finally we have 
\begin{align}
  &\left\| \nabla f(\x_{k+1}) \right\| \nonumber \\
 \leq  & \frac{L_2(1+\sqrt{1-\delta})^2\eta_k^2}{2(1-\beta)m} \sum_{i \in \mathcal{M}}\|\s_{i,k+1} \|^2 + \epsilon_g + \frac{M\gamma^2}{2(1-\alpha)m } \sum_{i \in \mathcal{M}}\|\s_{i,k+1} \|^2 \nonumber \\
&  +\left( \frac{(1-\alpha)}{(1-\beta)}2L_1(1+ \sqrt{1-\delta}) -\frac{\gamma}{\eta_k} (L_1- \epsilon_H)\right)\frac{1}{(1-\alpha)m} \sum_{i \in \mathcal{M}}\|\eta_k\s_{i,k+1}\|  \nonumber \\
& + \frac{L_2\alpha(1+\sqrt{1-\delta})^2\eta_k^2}{2(1-\beta)}\Gamma^2 + \frac{\eta_k\alpha}{(1-\beta)}(1 + \sqrt{1-\delta})\Gamma \nonumber 
\end{align}

Now we choose $\gamma>\frac{(1-\alpha)}{(1-\beta)}2L_1(1+ \sqrt{1-\delta}) \frac{\eta_k}{L_1-\epsilon_H}$.

\begin{align*}
    &\left\| \nabla f(\x_{k+1}) \right\| \\
  \leq  & \left[ \frac{L_2(1-\alpha)(1+\sqrt{1-\delta})^2}{2(1-\beta)} +\frac{M\gamma^2}{2\eta_k^2}\right]\frac{1}{(1-\alpha)m} \sum_{i \in \mathcal{M}}\|\eta_k\s_{i,k+1} \|^2 \\
  & + \frac{L_2\alpha(1+\sqrt{1-\delta})^2\eta_k^2}{2(1-\beta)}\Gamma^2 + \frac{\eta_k\alpha}{(1-\beta)}(1 + \sqrt{1-\delta})\Gamma +\epsilon_g \\
   \leq & \left[ \frac{L_2(1-\alpha)(1+\sqrt{1-\delta})^2}{2(1-\beta)} +\frac{M\gamma^2}{2\eta_k^2}\right]\left[\frac{1}{(1-\alpha)m} \sum_{i \in \mathcal{M}}\|\eta_k\s_{i,k+1} \|^3 \right]^{2/3} \\
  & + \frac{L_2\alpha(1+\sqrt{1-\delta})^2\eta_k^2}{2(1-\beta)}\Gamma^2 + \frac{\eta_k\alpha}{(1-\beta)}(1 + \sqrt{1-\delta})\Gamma +\epsilon_g
\end{align*}
At step $k=k_0$,
\begin{align*}
   & \left\| \nabla f(\x_{k_0+1}) \right\| \\
   &\leq \left[ \frac{L_2(1-\alpha)(1+\sqrt{1-\delta})^2}{2(1-\beta)} +\frac{M\gamma^2}{2\eta_k^2}\right]\left(\frac{\psi_{comp}}{T}\right)^{2/3} +\epsilon_g
   + \frac{L_2\alpha(1+\sqrt{1-\delta})^2\eta_k^2}{2(1-\beta)}\Gamma^2 + \frac{\eta_k\alpha}{(1-\beta)}(1 + \sqrt{1-\delta})\Gamma \\
   & \leq \frac{\chi_1}{T^{2/3}} +\epsilon_g + \text{Minor term of order } \mathcal{O}(\frac{1}{T}) 
\end{align*}

where $\chi_1= \left[ \frac{L_2(1-\alpha)(1+\sqrt{1-\delta})^2}{2(1-\beta)} +\frac{M\gamma^2}{2\eta_k^2}\right](\psi_{comp})^{2/3} $ and as $M = O(m)$, we have $\chi_1 = \mathcal{O}(m)$.

The Hessian bound is 
\begin{align}
 & \lambda_{\min} (\nabla^2 f(\x_{k+1}) ) \nonumber \\
 &= \frac{1}{(1-\alpha)m}\sum_{i \in \mathcal{M}} \lambda_{\min} \left[\nabla^2 f(\x_{k+1}) \right] \nonumber \\ 
  &= \frac{1}{(1-\alpha)m}\sum_{i \in \mathcal{M}} \lambda_{\min} \left[\bH_{i,k} - (\bH_{i,k}-\nabla^2 f(\x_{k+1}) )\right] \nonumber \\
  & \geq \frac{1}{(1-\alpha)m}\sum_{i \in \mathcal{M}} \left[\lambda_{\min}(\bH_{i,k}) - \|\bH_{i,k}-\nabla^2 f(\x_{k+1}) \|\right] \nonumber \\
  & \geq \frac{1}{(1-\alpha)m}\sum_{i \in \mathcal{M}}  \lambda_{\min}(\bH_{i,k}) - \frac{1}{(1-\alpha)m}\sum_{i \in \mathcal{M}}   \|\bH_{i,k}- \nabla^2 f(\x_{k+1}) \|   \nonumber \\
  & \geq \frac{1}{(1-\alpha)m}\sum_{i \in \mathcal{M}}  - \frac{M\gamma}{2}\|\s_{i,k+1}\| - \frac{1}{(1-\alpha)m}\sum_{i \in \mathcal{M}}  \|\bH_{i,k}- \nabla^2 f(\x_{k}) \| \nonumber \\ & \qquad  -\frac{1}{(1-\alpha)m}\sum_{i \in \mathcal{M}}   \|\nabla^2 f(\x_{k})- \nabla^2 f(\x_{k+1}) \| \nonumber  \\
  & \geq \frac{1}{(1-\alpha)m}\sum_{i \in \mathcal{M}}  - \frac{M\gamma}{2\eta_k}\|\eta_k\s_{i,k+1}\| - \epsilon_H - \frac{1}{(1-\alpha)m}\sum_{i \in \mathcal{M}}  L_2 \|\x_{k}- \x_{k+1} \| \nonumber  \\
  & \geq - \frac{M\gamma}{2\eta_k}\left[ \frac{1}{(1-\alpha)m}\sum_{i \in \mathcal{M}} \|\eta_k\s_{i,k+1}\|^3\right]^{1/3} -L_2 \|\x_{k}- \x_{k+1} \| -\epsilon_H \nonumber \\
  &\geq - \frac{M\gamma}{2\eta_k}\left[ \frac{1}{(1-\alpha)m}\sum_{i \in \mathcal{M}} \|\eta_k\s_{i,k+1}\|^3\right]^{1/3} -L_2 \left[\frac{(1+\sqrt{1-\delta})}{(1-\beta)m}\sum_{i\in \mathcal{U}}\|\eta_k\s_{i,k+1}\| \right] -\epsilon_H \nonumber \\
   &\geq - \frac{M\gamma}{2\eta_k}\left[ \frac{1}{(1-\alpha)m}\sum_{i \in \mathcal{M}} \|\eta_k\s_{i,k+1}\|^3\right]^{1/3} -L_2 \frac{(1+\sqrt{1-\delta})}{(1-\beta)m}\left[\sum_{i\in \mathcal{M}}\|\eta_k\s_{i,k+1}\| + \sum_{i\in \mathcal{B}\cap\mathcal{U}}\|\eta_k\s_{i,k+1}\|\right] -\epsilon_H \nonumber \\
   &\geq - \frac{M\gamma}{2\eta_k}\left[ \frac{1}{(1-\alpha)m}\sum_{i \in \mathcal{M}} \|\eta_k\s_{i,k+1}\|^3\right]^{1/3} -L_2 \frac{(1+\sqrt{1-\delta})(1-\alpha)}{(1-\beta)}\left[\frac{1}{(1-\alpha)m}\sum_{i\in \mathcal{M}}\|\eta_k\s_{i,k+1}\| \right] -\epsilon_H \nonumber\\ & -L_2 \frac{(1+\sqrt{1-\delta})\alpha}{(1-\beta)}\eta_k\Gamma
\end{align}
At $k=k_0$ we have 
\begin{align}
  &\lambda_{\min} (\nabla^2 f(\x_{k_0+1}) )\\ \geq &  -\frac{M\gamma}{2\eta_k}\left[ \frac{\psi_{comp}}{T}  \right]^{1/3} -L_2 \frac{(1+\sqrt{1-\delta})(1-\alpha)}{(1-\beta)}\left[ \frac{\psi_{comp}}{T}  \right]^{1/3}-\epsilon_H -L_2 \frac{(1+\sqrt{1-\delta})\alpha}{(1-\beta)}\eta_k\Gamma\nonumber \\
  \geq &-\left[\frac{M\gamma}{2\eta_k} + L_2 \frac{(1+\sqrt{1-\delta})(1-\alpha)}{(1-\beta)}\right]\psi^{1/3}_{comp} \left(\frac{1}{T} \right)^{1/3}-\epsilon_H -\text{ Minor term }\mathcal{O}(1/T)\nonumber \\
  \geq &  -\frac{\chi_2}{T^{1/3}} -\epsilon_H-\text{ Minor term }\mathcal{O}(1/T)
 \end{align}
 where $\chi_2 =\left[\frac{M\gamma}{2\eta_k} + L_2 \frac{(1+\sqrt{1-\delta})(1-\alpha)}{(1-\beta)}\right]\psi^{1/3}_{comp}  $ and we have $\chi_2= \mathcal{O}(m)$.

 We have the following parameters
 \begin{align}
  \psi_{comp} &=\left[ \frac{f(\x_{0}) -f^*}{\lambda_{comp}}+ \sum_{k=0}^{T-1}\frac{\lambda_{\Gamma}}{\lambda_{comp}}  \right]   \label{eq:pcomp} 
 \end{align}

%% file: appendixexp_ai.tex
\section{Additional Experiments}
In this section we provide additional experiments.  We choose the parameters $\lambda=1, M=10$, learning rate $\eta_k=1$, fraction of the Byzantine machines $\alpha=\{.1,.15,.2\}$ and $\beta= \alpha + \frac{2}{m}$. 

\paragraph*{Compressed and Byzantine:} In Figure ~\ref{fig:comprobustbyz1}, we plot the function value of the robust linear regression problem for 'Gaussian attack` and 'random label` attacks with compressed update for both `w8a' and `a9a' dataset.

\begin{figure}[h!]%
\centering
\subfigure[]{
\includegraphics[height = 3.5cm,width=3.5cm]{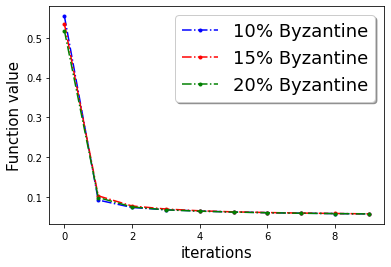}}%
\subfigure[]{
\includegraphics[height = 3.5cm,width=3.5cm]{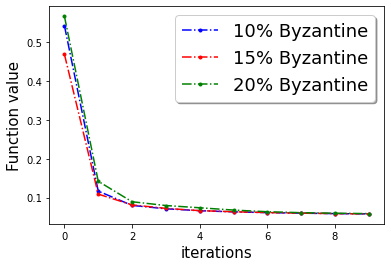}}%
\subfigure[]{
\includegraphics[height = 3.5cm,width=3.5cm]{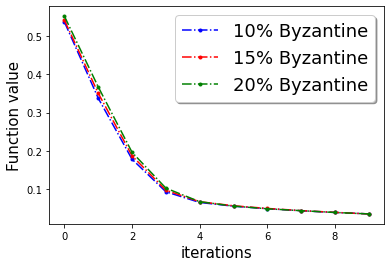}}%
\subfigure[]{
\includegraphics[height = 3.5cm,width=3.5cm]{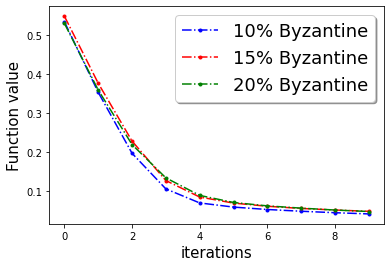}}%
\caption{Loss  of the training data `a9a' (first row) and `w8a' (second row) with $10\%,15\%,20\% $ Byzantine worker machines for (a,c). Gaussian attack.(b,d).Random attack.}%
    \label{fig:comprobustbyz1}%
    \vspace{-4mm}
\end{figure}

\begin{figure}[h]%
\centering
\subfigure[]{
\includegraphics[height = 3.5cm,width=3.5cm]{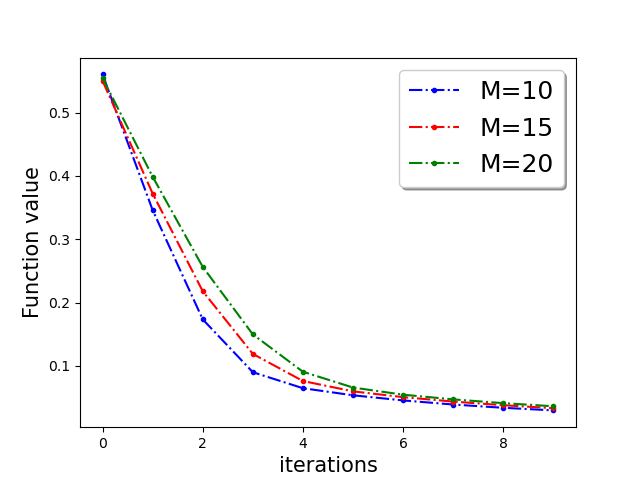}}%
\subfigure[]{
\includegraphics[height = 3.5cm,width=3.5cm]{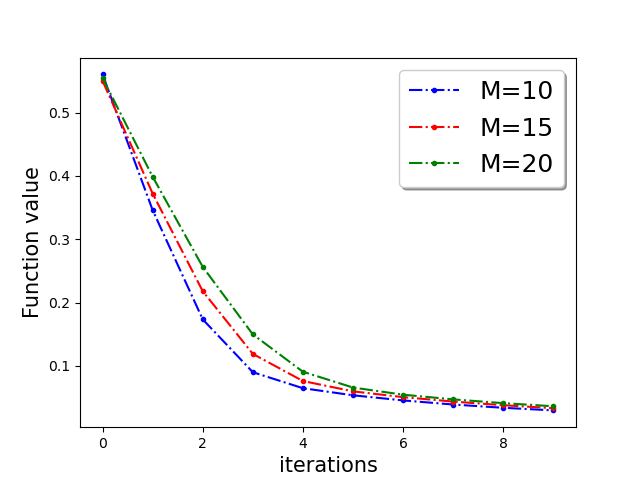}}%
\subfigure[]{
\includegraphics[height = 3.5cm,width=3.5cm]{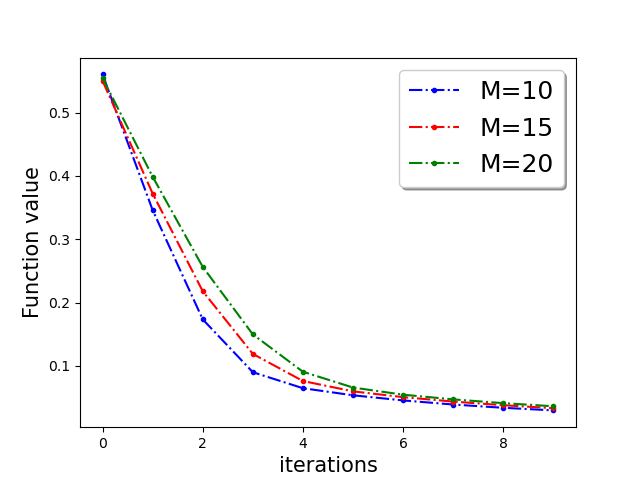}}%
\subfigure[]{
\includegraphics[height = 3.5cm,width=3.5cm]{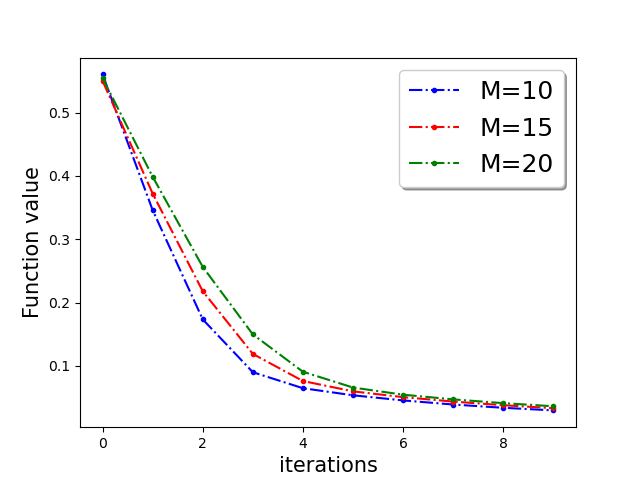}}%
\\
\subfigure[]{
\includegraphics[height = 3.5cm,width=3.5cm]{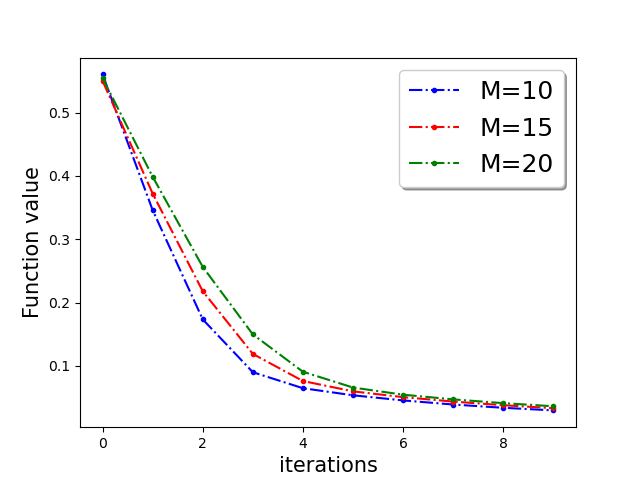}}%
\subfigure[]{
\includegraphics[height = 3.5cm,width=3.5cm]{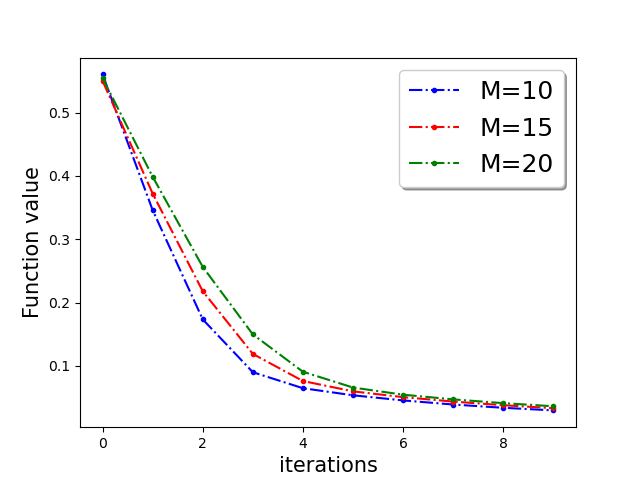}}%
\subfigure[]{
\includegraphics[height = 3.5cm,width=3.5cm]{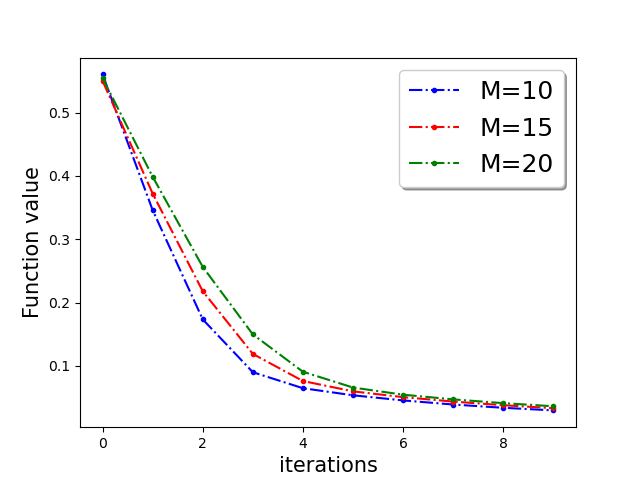}}%
\subfigure[]{
\includegraphics[height = 3.5cm,width=3.5cm]{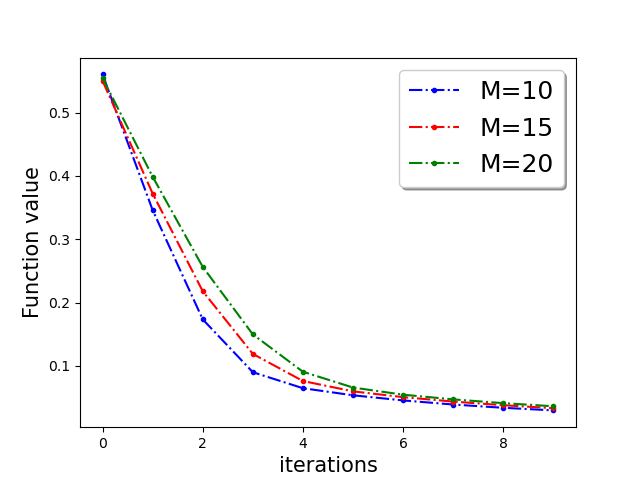}}%
\caption{Function loss  of the training data `a9a' dataset (first row) and `w8a' dataset (second row) with $10\%,15\%,20\% $ Byzantine worker machines for (a,e). Flipped label attack.(b,f). Negative Update attack (c,g). Gaussian noise attack and (d,h). Random label attack for non-convex robust linear regression problem.}%
    \label{fig:robustbyz2}%
    \vspace{-6mm}
\end{figure}

\begin{figure}[h]%
\centering
\subfigure[]{
\includegraphics[height = 3.5cm,width=3.5cm]{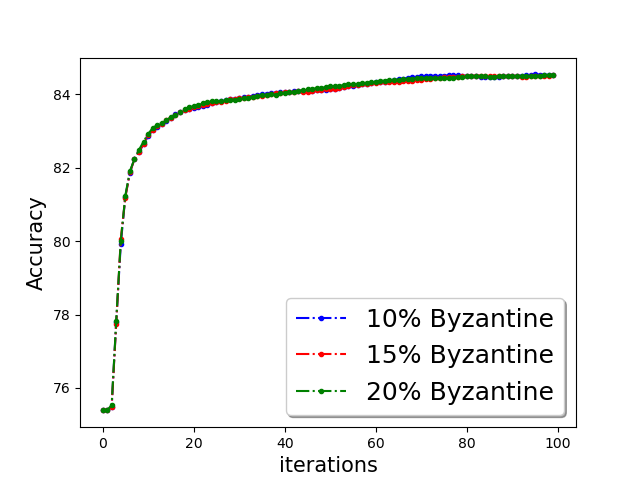}}%
\subfigure[]{
\includegraphics[height = 3.5cm,width=3.5cm]{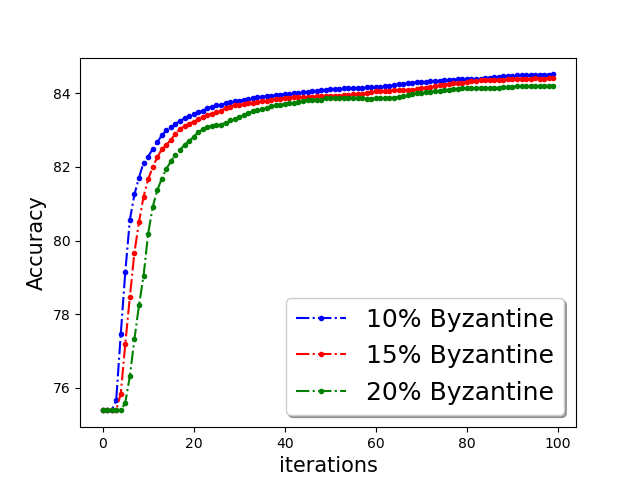}}%
\subfigure[]{
\includegraphics[height = 3.5cm,width=3.5cm]{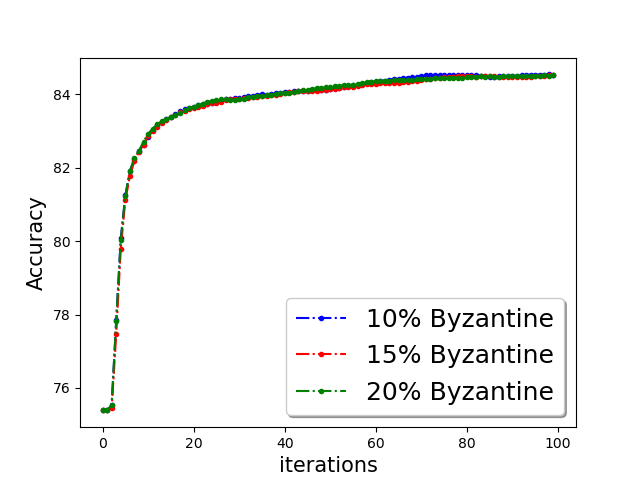}}%
\subfigure[]{
\includegraphics[height = 3.5cm,width=3.5cm]{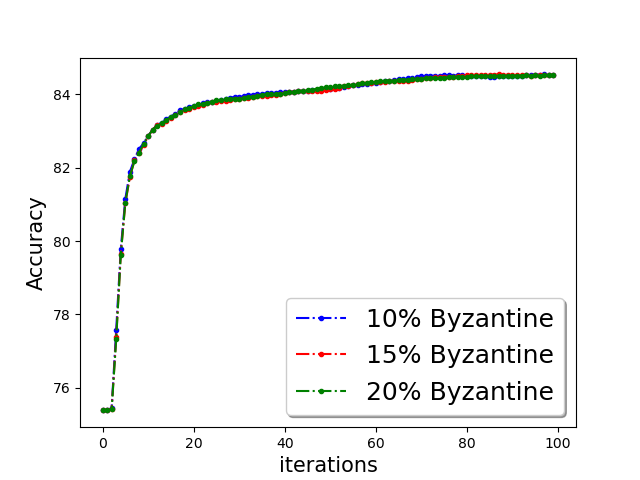}}%
\\
\subfigure[]{
\includegraphics[height = 3.5cm,width=3.5cm]{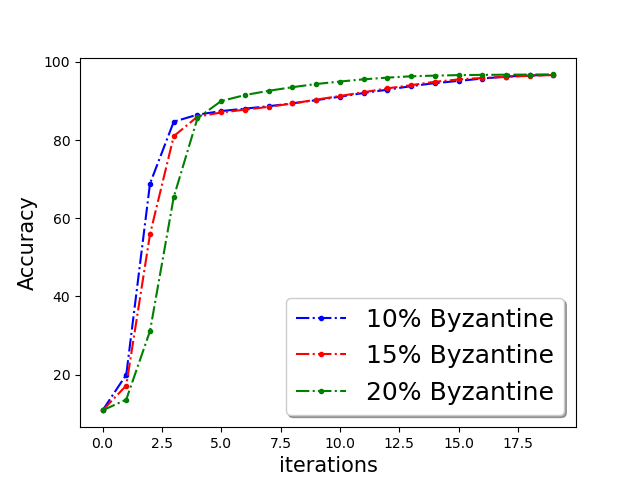}}%
\subfigure[]{
\includegraphics[height = 3.5cm,width=3.5cm]{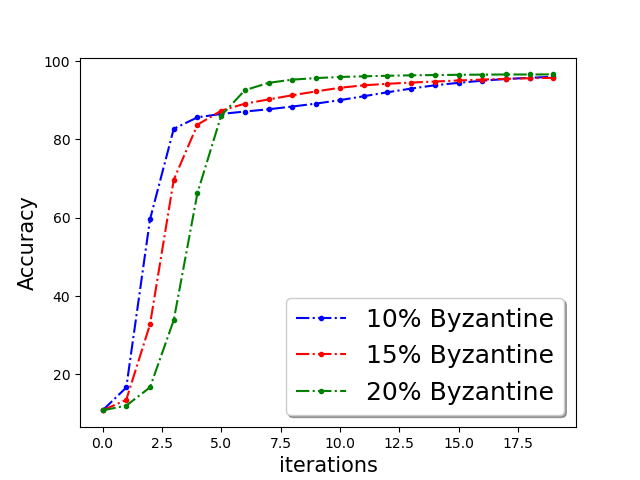}}%
\subfigure[]{
\includegraphics[height = 3.5cm,width=3.5cm]{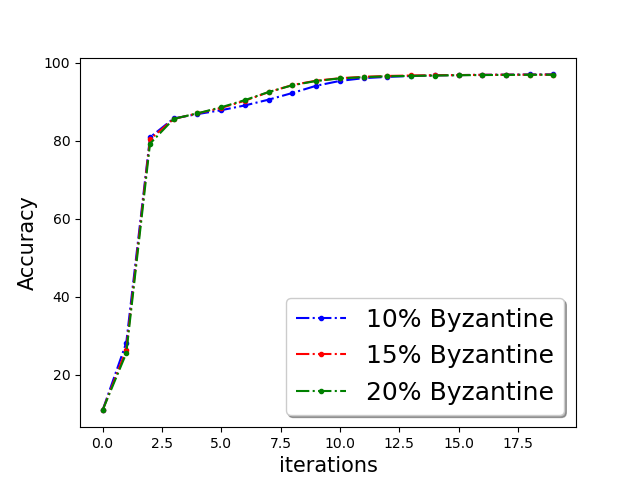}}%
\subfigure[]{
\includegraphics[height = 3.5cm,width=3.5cm]{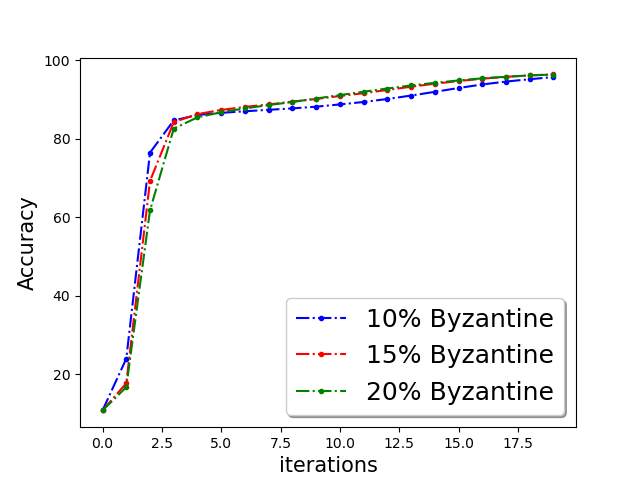}}%
\caption{Classification accuracy of the testing data `a9a' dataset (first row) and `w8a' dataset (second row) with $10\%,15\%,20\% $ Byzantine worker machines for (a,e). Flipped label.(b,f). Negative Update  (c,g). Gaussian noise and (d,h). Random label attack for logistic regression problem. }%
    \label{fig:logbyz}%
    \vspace{-2mm}
\end{figure}

\paragraph*{Training loss for uncompressed update:} In Figure ~\ref{fig:robustbyz2}, we plot the function value of the robust linear regression problem for all four  attacks update for both `w8a' and `a9a' dataset with uncompressed update $(\delta=1)$.

\paragraph*{Classification accuracy: } We show the classification accuracy on testing data of `a9a' and `w8a' dataset for logistic regression problem in Figure ~\ref{fig:logbyz} 
and   training function loss of `a9a' and `w8a' dataset for robust linear regression problem in the Figure~\ref{fig:logbyz}. It is evident from the plots that a simple \emph{norm based thresholding} makes the learning algorithm robust.